\documentclass[letterpaper, 10 pt, conference]{ieeeconf}%
\usepackage{graphicx}
\usepackage{epsfig}
\usepackage{amsmath}
\usepackage{amssymb}
\usepackage{bbm}
\usepackage{float}
\usepackage{tikz}
\usepackage{algorithm}
\usepackage{algpseudocode}
\usepackage{pgfplots}
\usepackage{gnuplottex}
\usepackage{mathrsfs}
\usepackage{subfigure}
\usepackage{booktabs}
\usepackage{multirow}
\usepackage{amsfonts}%
\providecommand{\U}[1]{\protect\rule{.1in}{.1in}}
\providecommand{\U}[1]{\protect\rule{.1in}{.1in}}
\IEEEoverridecommandlockouts
\newcommand\Square[1]{+(-#1,-#1) rectangle +(#1,#1)}
\makeatletter
\let\@ORGmakecaption\@makecaption
\long
\def\@makecaption#1#2{\@ORGmakecaption{#1}{#2}\vskip\belowcaptionskip\relax}
\makeatother
\def\@normalsize{\@setsize\normalsize{10pt}\xpt\@xpt
\abovedisplayskip 10pt plus2pt minus5pt\belowdisplayskip
\abovedisplayskip \abovedisplayshortskip \z@
plus3pt\belowdisplayshortskip 6pt plus3pt
minus3pt\let\@listi\@listI}
\allowdisplaybreaks
\makeatletter
\newcommand\npar{\@startsection{section}{1pt}}
\makeatother
\DeclareMathOperator*{\argmin}{\arg\!\min}

\newcommand{\bsym}[1]{\boldsymbol{#1}}
\newtheorem{theorem}{Theorem}
\newcounter{lemcount}
\newenvironment{lemma}[1][]{\refstepcounter{lemcount}\par  \it{Lemma~\thelemcount. #1 }
\rmfamily  }{\par}

\newcommand{\qed}{\nobreak \ifvmode \relax \else
\ifdim\lastskip<1.5em \hskip-\lastskip
\hskip1.5em plus0em minus0.5em \fi \nobreak
\vrule height0.4em width0.5em depth0.25em\fi}
\title{{\LARGE \textbf{A New Event-Driven Cooperative Receding Horizon Controller for
Multi-agent Systems in Uncertain Environments}}}
\author{ \parbox{3.5 in}{\centering Yasaman Khazaeni and Christos G. Cassandras\\
         Division of Systems Engineering\\ and Center for Information and Systems Engineering\\
         Boston University\\
         Brookline, MA 02446\\
         {\tt\small yas@bu.edu,cgc@bu.edu}}
         \thanks{The authors' work is supported in part by NSF under Grant CNS-1139021, by AFOSR under grant FA9550-12-1-0113, by ONR under grant N00014-09-1-1051, and by ARO under Grant W911NF-11-1-0227.}
}
\begin{document}
\maketitle
\thispagestyle{empty}
\pagestyle{empty}
\begin{abstract}
In previous work, a Cooperative Receding Horizon (CRH) controller was developed for solving cooperative multi-agent problems in uncertain environments. In this paper, we overcome several limitations of this controller, including potential instabilities in the agent trajectories and poor performance due to inaccurate estimation of a reward-to-go function. We propose an event-driven CRH controller to solve the maximum reward collection problem (MRCP) where multiple agents cooperate to maximize the total reward collected from a set of stationary targets in a given mission space. Rewards are non-increasing functions of time and the environment is uncertain with new targets detected by agents at random time instants. The controller sequentially solves optimization problems over a planning horizon and executes the control for a shorter action horizon, where both are defined by certain events associated with new information becoming available. In contrast to the earlier CRH controller, we reduce the originally infinite-dimensional feasible control set to a finite set at each time step. We prove some properties of this new controller and include simulation results showing its improved performance compared to the original one.
\end{abstract}

\section{Introduction}

Cooperative control is used in systems where a set of control agents with
limited sensing, communication and computational capabilities seeks to achieve
objectives defined globally or individually \cite{Shamma2007},\cite{Murphey2002}.
Uncertain environments further require the agents to respond to random events.
Examples arise in UAV teams, cooperative classification, mobile agent
coordination, rendez-vous problems, task assignment, persistent monitoring, coverage control
and consensus problems; see \cite{Jadbabaie2003},\cite{McLain2001-p},\cite{CGC10},\cite{CGC2013_2},\cite{Cortes2004},\cite{Zhong2011},\cite{Panagou2014},\cite{Ren2008},\cite{Zhong2010}
and references therein. Both centralized and decentralized control approaches
are used; in the latter case, communication between the agents in order to
make collaborative decisions plays a crucial role.

In this paper, we consider Maximum Reward Collection Problems (MRCP) where $N$
agents are collecting time-dependent rewards associated with $M$ targets in an
uncertain environment. In a deterministic setting with equal target rewards a
one-agent MRCP is an instance of a Traveling Salesman Problem (TSP),
\cite{Salz1966},\cite{applegate2011}. The multi-agent MRCP is similar to the
Vehicle Routing Problem (VRP) \cite{Laporte1992}. These are combinatorial problems
for which globallly optimal solutions are found through integer programming.
For example, in \cite{Ekici2013},\cite{Tang2007}, a deterministic MRCP with a
linearly decreasing reward model is cast as a dynamic scheduling problem and
solved via heuristics.

Because of the MRCP complexity, it is natural to resort to decomposition
techniques. One approach is to seek a functional decomposition that divides
the problem into smaller sub-problems \cite{Bellingham02},\cite{Earl07_2} which may
be defined at different levels of the system dynamics. An alternative is a
time decomposition where the main idea is to solve a finite horizon
optimization problem, then continuously extend this \emph{planning horizon}
forward (either periodically or in event-driven fashion). This is in the
spirit of receding horizon techniques used in Model Predictive Control (MPC)
to solve optimal control problems for which obtaining infinite horizon
feedback control solutions is extremely difficult \cite{Mayne2000}. In such
methods, the current control action is calculated by solving a finite horizon
open-loop optimal control problem using the current state of the system as the
initial state. At each instant, the optimization yields an optimal control
sequence executed over a shorter \emph{action horizon} before the process is
repeated. In the context of multi-agent systems, a Cooperative Receding
Horizon (CRH) controller was introduced in \cite{Li2006} with the controller
steps defined in event-driven fashion (with events dependent on the observed
system state) as opposed to being invoked periodically, in time-driven
fashion. A key feature of this controller is that it does not attempt to make
any explicit agent-to-target assignments, but only to determine headings that,
at the end of the current planning horizon, would place agents at positions
such that a total expected reward is maximized. Nonetheless, as shown in
\cite{Li2006}, a stationary\ trajectory for each agent is guaranteed under
certain conditions, in the sense that an agent trajectory always converges to
some target in finite time.

In this paper, we consider MRCPs in uncertain environments where, for
instance, targets appear/disappear at random times and a target may have a
random initial reward and a random reward decreasing rate. The contribution is
to introduce a new CRH controller, allowing us to overcome several limitations
of the controller in \cite{Li2006}, including potential instabilities in the
agent trajectories and poor performance due to inaccurate estimation of the
reward-to-go function. We accomplish this by reducing, at each event-driven
control evaluation step, the originally infinite-dimensional feasible control
set to a finite set and by improving the estimation process for the reward to
go, including a new \textquotedblleft travel cost factor\textquotedblright%
\ for each target which accommodates different target configurations in a
mission space. We also establish some properties of this new controller whose
overall performance is significantly better relative to the original one, as
illustrated through various simulation examples.

In section \ref{ProbForm}, the MRCP is formulated and in Section \ref{MDP} we
place the problem in a broader context of event-driven optimal control. In
Sections \ref{origCRH} and \ref{NewCRH} the original CRH controller is
reviewed and the proposed new controller and some of its properties are
established. In Section \ref{Numerical} simulation examples are presented and
future research directions are outlined in the conclusions.
\section{Problem Formulation}

\label{ProbForm} We consider a MRCP where agents and targets are located in a
mission space $S$. There are $M$ targets defining a set $\mathcal{T}=\{1,..,M\}$
and $N$ agents defining a set $\mathcal{A}=\{1,...,N\}$. The mission space may
have different topological characteristics. In a Euclidean topology,
$S\subset\mathbb{R}^{2}$ as illustrated in Fig. \ref{ObstacleMission} with a
triangle denoting a base, circles are agents and squares are targets. In this
case, the distance metric $d(x,y)$ is a simple Euclidean norm such that
$\mathit{d}:S\times S\rightarrow\mathbb{R}$ is the length of the shortest path
between points $x,y\in S$. Moreover, the feasible agent headings are given by
the set $\mathbf{U}_{j}(t)=[0,2\pi],$ $j\in\mathcal{A}$. If there are
obstacles in $S$, the feasible headings and the shortest path between two
points should be defined accordingly. Alternatively, the mission space may be
modeled as a graph $\mathcal{G}(E,V)$ with $V$ representing the location of
targets and the base. Feasible headings are defined by the (directed) edges at
each node and the distance $d(u,v)$ is the sum of the edge weights on the
shortest path between $u$ and $v$. In this paper, we limit ourselves to a
Euclidean mission topology.

Targets are located at points $\mathbf{y}_{i}\in S$, $i\in\mathcal{T}$. Target
$i$'s reward is denoted by $\lambda_{i}\phi_{i}(t)$ where $\lambda_{i}$ is the
initial maximum reward and $\phi_{i}(t)\in\lbrack0,1]$ is a non-increasing
discount function. By using the appropriate discounting function we can
incorporate constraints such as hard or soft deadlines for targets. An example
of a discount function is
\begin{equation}
\phi_{i}(t)=\left\{
\begin{array}
[c]{ll}%
1-\frac{\alpha_{i}}{D_{i}}t & \mbox{  if $t \le D_i$}\\
(1-\alpha_{i})e^{-\beta_{i}(t-D_{i})} & \mbox{   if $t>D_i$}
\end{array}
\right.  \label{reward}%
\end{equation}
where $\alpha_{i}$, $\beta_{i}$ and $D_{i}$ are given parameters. Agents are
located at $\mathbf{x}_{j}(t)\in S$. Each agent has a controllable heading at
time $t$, $u_{j}(t)\in\mathbf{U}_{j}(t)=[0~2\pi]$. The velocity of agent $j$
is
\begin{equation}
\mathbf{v}_{j}(t)=V_{j}%
\begin{bmatrix}
\cos(u_{j}(t))\\
\sin(u_{j}(t))
\end{bmatrix}
\label{velocity}%
\end{equation}
where we assume that $V_{j}$ is a fixed speed.

We define a \emph{mission} as the process of the $N$ agents cooperatively
collecting the maximum possible total reward from $M$ targets within a given
mission time $T$. Upon collecting rewards from all targets, the agents deliver
it to a \emph{base} located at $\mathbf{z}\in S$. Events occurring during a
mission can be controllable (e.g., collecting a target's reward) or random
(e.g., the appearance/disappearance of targets or changes in their location).
The event-driven CRH controller we will develop, handles these random events by re-solving the optimal control problem as in the original CRH controller in
\cite{Li2006}. In order to ensure that agents collect target rewards in finite
time, we assume that each target has a radius $s_{i}>0$ and that agent $j$
collects reward $i$ at time $t$ if and only if $\mathit{d}(\mathbf{x}%
_{j}(t),\mathbf{y}_{i})\leq s_{i}$. 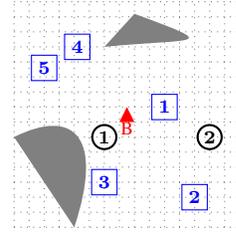
\begin{figure}[ptb]
\centering
\begin{tikzpicture}[scale=0.4]
\draw[thin,gray,dotted,step=.5] (0,0) grid (7.5,7.5);
\draw[black,thick](3,3) circle (0.4);
\draw[black,thick](6.5,3) circle (0.4);
\draw[black](3,3) node{\scriptsize ${\bf 1}$};
\draw[black](6.5,3) node{\scriptsize ${\bf 2}$};
\color{blue}
\draw(5,4) \Square{12pt};
\draw(6,1) \Square{12pt};
\draw(3,1.5) \Square{12pt};
\draw(2.1,6) \Square{12pt};
\draw(1,5.3) \Square{12pt};
\draw[blue](5,4) node{\scriptsize ${\bf 1}$};
\draw[blue](6,1) node{\scriptsize ${\bf 2}$};
\draw[blue](3,1.5) node{\scriptsize ${\bf 3}$};
\draw[blue](2.1,6) node{\scriptsize ${\bf 4}$};
\draw[blue](1,5.3) node{\scriptsize ${\bf 5}$};
\fill[red](3.5,3.5) -- (3.75,4) -- (4,3.5) -- cycle;
\draw[red](3.75,3.3) node{\scriptsize B};
%obstacles
\fill[gray] (2,0).. controls (3,3) and (2,4) .. (0,3);
\fill[gray] (3,6).. controls (6,6.3) and (7,6) .. (4,7.1);
\end{tikzpicture}
\caption{Sample mission space with filled blue regions as obstacles}%
\label{ObstacleMission}%
\end{figure}

\section{An Event-Driven Optimization View}

\label{MDP} We view the solution of a MRCP as a sequence of headings for all
agents and associated heading switching times. We define a policy
${\boldsymbol{\pi}}$ as a vector $[\mathbf{u},{\xi}]$ where ${\xi}=[\xi
_{1},...,\xi_{K}]$ are the switching time intervals over which headings are
maintained with $t_{k+1}=\sum_{l=1}^{k}\xi_{l}$, and $t_{1}=0$. The control
$\mathbf{u}=[\mathbf{u}_{1},...,\mathbf{u}_{K}]$ with $\mathbf{u}_{k}%
=[u_{1}(t_{k}),...,u_{N}(t_{k})]$ is the vector of all the agent headings at
time $t_{k}$. With $M$ bounded, there exist policies ${\boldsymbol{\pi}}$ such
that all targets are visited over a finite number of switching events. Each
switching time $t_{k}$ is either the result of a controllable event (e.g.,
visiting a target) or an uncontrollable random event. This is a complex
stochastic control problem where the state space $\Xi$ is the set of
all possible location of agents $\mathcal{X}_{k}=[\mathbf{x}_{1}%
(t_{k}),...,\mathbf{x}_{N}(t_{k})]$ and targets $\mathcal{Y}_{k}%
=[\mathbf{y}_{1},...,\mathbf{y}_{M_{k}}]$ with $M_{k}=\Vert\mathcal{T}%
_{k}\Vert$ and $\mathcal{T}_{k}$ is the set of unvisited targets at time
$t_{k}$. As the mission evolves, $M_{k}$ decreases and the mission is complete
when either $M_{k}=0$ or a given mission time $T$ is reached. The complete
system state at time $t_{k}$ is $(\mathcal{X}_{k},\mathcal{Y}_{k}%
)\in\Xi$. We define the optimization problem $\mathbf{P}$ as:
\begin{equation}
\max\limits_{{\boldsymbol{\pi}}}\sum_{k=1}^{K}R_{\boldsymbol{\pi}}%
(t_{k},\mathcal{X}_{k},\mathcal{Y}_{k}) \label{ProblemP}%
\end{equation}
where
\[
R_{\boldsymbol{\pi}}(t_{k},\mathcal{X}_{k},\mathcal{Y}_{k})=\sum_{i=1}^{M_{k}%
}\sum_{j=1}^{N}\lambda_{i}\phi_{i}(t_{k})\mathbbm{1}\{\mathit{d}%
(\mathbf{x}_{j}(t_{k}),\mathbf{y}_{i})\leq s_{i}\}
\]
The time a target is visited is a controllable event associated with a heading
switching. In a deterministic problem, there is no need to switch headings
unless a target is visited, but in an uncertain setting the switching times
are not limited to these events. We define a subsequence ${\boldsymbol{\tau}%
}^{\boldsymbol{\pi}}=\{\tau_{1}^{\boldsymbol{\pi}},\tau_{2}^{\boldsymbol{\pi}%
},...,\tau_{M}^{\boldsymbol{\pi}}\}$ of $\{t_{1},...t_{K}\}$, $M\leq K,$ so
that $\tau_{i}^{\boldsymbol{\pi}}$ is the time target $i$ is visited. Note
that $\boldsymbol{\tau}^{\boldsymbol{\pi}}$ is not a monotonic sequence, since
targets can be visited in any order.
%Let $\sigma_{\boldsymbol{\pi}}$ be a
%permutation of $\{1,2,..,M\}$ with the order in which targets are visited
%under policy $\boldsymbol{\pi}$. We can sort ${\boldsymbol{\tau}}$ into
%$\{\tau_{\sigma_{\boldsymbol{\pi}}(1)},\tau_{\sigma_{\boldsymbol{\pi}}%
%(2)},...,\tau_{\sigma_{\boldsymbol{\pi}}(M)}\}$ to be the ordered switching
%times at which one target is visited.
Therefore, \eqref{ProblemP} can be rewritten as
\begin{equation}
\max\limits_{\boldsymbol{\pi}}\sum_{i=1}^{M}\lambda_{i}\phi_{i}(\tau
_{i}^{\boldsymbol{\pi}}) \label{tauformulation}%
\end{equation}
Defining the immediate reward as being collected during a time period $\xi
_{k}$ and the reward-to-go as being aggregated over all $t>t_{k}+\xi_{k}$, an
optimality equation for this problem is:
\begin{equation}
\resizebox{1.0 \columnwidth}{!}{$J(t_k,\mathcal{X}_k,\mathcal{Y}_k)=\max\limits_{{\bf u}_k,\xi_k}\big[J_I(t_k,\mathcal{X}_k,\mathcal{Y}_k,{\bf u}_k,\xi_k)+J(t_{k+1},\mathcal{X}_{k+1},\mathcal{Y}_{k+1})\big]$}
\label{optequation}%
\end{equation}
where $J(t_{k},\mathcal{X}_{k},\mathcal{Y}_{k})$ denotes the maximum total
reward at time $t_{k}$ with current state $(\mathcal{X}_{k},\mathcal{Y}_{k})$
and $J_{I}(t_{k},\mathcal{X}_{k},\mathcal{Y}_{k},\mathbf{u}_{k},\xi_{k})$ is
the immediate reward collected in the interval $(t_{k},t_{k+1}]$. Finally,
$J(t_{k+1},\mathcal{X}_{k+1},\mathcal{Y}_{k+1})$ is the maximum reward-to-go at $t_{k+1}$ assuming no future uncertainty, i.e., we avoid the use of an {\it a priori} stochastic model for the environment, opting instead to react to random
events by re-solving \eqref{optequation} when this happens. Letting $\tau^*=\max_{i\in\mathcal{T}}\{\tau_i^{\bsym \pi}\}$, we set $J(\tau^*,\mathcal{X}_{K},\mathcal{Y}_K)=0$. Henceforth, we write $J(t_{k},\mathcal{X}_{k},\mathcal{Y}_{k}%
)=J(t_{k})$ for brevity.
Had we assumed a fixed value for $\xi_{k}$ \textit{a priori}, the optimization problem \eqref{optequation} could have been solved using Dynamic
Programming (DP) with the terminal state reached when no target is left in the mission space. However, a fixed $\xi_k$ does not allow for real-time reactions to new events. This fact, along with the size of the state space renders DP impractical and motivates a
receding horizon control approach where we set $\xi_{k}=H_{k}$ based on a
\emph{planning horizon} $H_{k}$ selected at time step $t_{k}$. Then, a finite horizon optimal control problem over
$(t_{k},t_{k}+H_{k}]$ is solved to determine the optimal control $\mathbf{u}^*_{k}$.
This control is maintained for an \emph{action horizon} $h_{k}\leq H_{k}$. A
new optimization problem is re-solved at $t_{k+1}=t_{k}+h_{k}$ or earlier if
any random event is observed. Following \eqref{optequation}, the optimization
problem $\mathbf{P}_{k}$ is
\begin{equation}
\max_{\mathbf{u}_{k}}[J_{I}(\mathbf{u}_{k},t_{k},H_{k}%
)+J(t_{k+1},H_{k+1})] \label{dp_problem}%
\end{equation}
where $J(t_{k+1},H_{k+1})$ and $J_{I}(\mathbf{u}_{k},t_{k},H_{k})$ were defined
above assuming $\xi_{k}=H_{k}$. The immediate reward
is zero if agents do not visit any target during $(t_{k},t_{k}+H_{k}]$,
otherwise it is the reward collected over this interval. Fixing the value of
$H_{k}$ is not constraining, since it is always possible to stop and re-solve
a new problem at any $t>t_{k}$.
\section{CRH Control Scheme}

\label{origCRH} In this section we briefly review the CRH controller
introduced in \cite{Li2006} and identify several limitations of it to motivate
the methods we will use use to overcome them.

{\it Cooperation Scheme:} In \cite{Li2006} the agents divide the mission space into a
\emph{dynamic} partition at each mission step. The degree of an agent's
responsibility for each target depends on the relative proximity of the agent
to the target. A neighbor set is defined for each target which includes its
$b$ closest agents, $b=1,2,\ldots$, sharing the responsibility for that target
until another agent moves closer. A value of $b=2$ is used in the previous and
current work for simplicity. Defining $c_{ij}(t)=\mathit{d}(\mathbf{y}%
_{i},\mathbf{x}_{j}(t))$ to be the direct distance between target $i$ and
agent $j$ at time $t$, let $B^{l}(\mathbf{y}_{i},t)$ be the $l$th closest
agent to target $i$ at time $t$. Formally,
%\begin{align}%
\begin{equation}
B^{l}(i,t)=\argmin\limits_{j\in\mathcal{A},j\neq B^{1}(i,t),...,j\neq
B^{l-1}(i,t)}\{c_{ij}(t)\}
\end{equation}
%\end{align}
Let $\beta^{b}(i,t)=\{B^{1}(i,t),...,B^{b}(i,t)\}$ be a neighbor set based on
which a \emph{relative distance} function is defined for all $j\in\mathcal{A}%
$:
\begin{equation}
\delta_{ij}(t)=\left\{
\begin{array}
[c]{ll}%
{\displaystyle\frac{c_{ij}(t)}{\sum\limits_{k\in\beta^{b}(i,t)}c_{ik}(t)}} &
\mbox{if $j \in \beta^b(i,t)$};\\
1 & \mbox{otherwise}
\end{array}
\right.
\end{equation}
Obviously, if $j\notin\beta^{b}(i,t)$, then $\delta_{ij}(t)=1$. The
\emph{relative proximity function} $p(\delta_{ij}(t))$ defined in
\cite{Li2006} is viewed as the probability that target $i$ will be visited by
agent $j$:
\begin{equation}
p(\delta_{ij}(t))=\left\{
\begin{array}
[c]{ll}%
1, & \mbox{if }\delta\leq\Delta\\
\frac{1-\Delta-\delta}{1-2\Delta}, & \mbox{if }\Delta\leq\delta\leq1-\Delta\\
0, & \mbox{if }\delta>1-\Delta
\end{array}
\right.  \label{proximity}%
\end{equation}
Here, $\Delta\in\lbrack0,\frac{1}{2})$ defines the level of cooperation
between the agents. By increasing $\Delta$ an agent will take full
responsibility for more targets, hence less cooperation. Each agent takes on
full responsibility for target $i$ if $\delta_{ij}(t)\leq\Delta$. As shown in
\cite{Li2006}, when $\Delta=\frac{1}{2}$ the regions converge to the Voronoi
tessellation of the mission space, with the location of agents at the centers
of the Voronoi tiles. There is no cooperation region in this case and each
agent is fully responsible for the targets within its own Voronoi tile. On the
other hand, when $\Delta=0$ no matter how close an agent is to a target, the
two agents are still responsible for that target.

{\it Planning and Action Horizons:} In \cite{Li2006}, $H_{k}$ is defined as
the earliest time of an event such that one of the agents can visit one of the
targets:
\begin{equation}
H_{k}=\min\limits_{l\in\mathcal{T}_{k}}\Big\{\frac{\mathit{d}(\mathbf{x}_{j}%
(t_{k}),\mathbf{y}_{l})}{V_{j}}\Big\} \label{HkDef}%
\end{equation}
This definition of \emph{planning horizon} for the CRH controller ensures no
controllable event can take place during this horizon. It also ensures that
re-evaluation of the CRH\ control is event-driven, as opposed to being
specified by a clock which involves a tedious synchronization over agents.
Fig. \ref{activetarget} illustrates how $H_{k}$ is determined when $V_{j}=1$.
The CRH control calculated at $t_{k}$ is maintained for an \emph{action
horizon} $h_{k}\leq H_{k}$. In \cite{Li2006} $h_{k}$ is defined either $(i)$
through a random event that may be observed at $t_{e}\in(t_{k},t_{k}+H_{k}]$
so that $h_{k}=t_{e}-t_{k}$, or $(ii)$ as $h_{k}=\gamma H_{k}$, $\gamma
\in(0,1)$. It is also shown in \cite{Li2006} that under (\ref{HkDef}) the CRH
controller generates a stationary trajectory for each agent under certain
conditions, in the sense that an agent trajectory always converges to some target in finite time.
\subsection{Original CRH Controller Limitations}

\emph{Instabilities in agent trajectories}\textit{:} The optimization problem
considered in \cite{Li2006} uses a potential function which is minimized in
order to maximize the total reward. The stationary\ trajectory guarantee
mentioned above is based on the assumption that all minima of this function
are at the target locations. If this assumption fails to hold, the agents are
directed toward the weighted center of gravity of all targets. This can happen
in missions where targets attain a symmetric configuration, leading to
oscillatory behavior in the agent trajectories. An example is shown in Fig.
\ref{OldCRH3target} with the original CRH controller applied to a single
agent, resulting in oscillations between three targets with equal rewards.
This problem was addressed in \cite{li2006centralized} by introducing a
monotonically increasing cost factor (or penalty) $C(u_{j})$ on the heading
$u_{j}$. While this prevents some of the instabilities, it has to be
appropriately tuned for each mission. We show how to overcome this problem in Section \ref{NewCRH}.

\emph{Hedging and mission time}\textit{:} The agent trajectories in
\cite{Li2006} are specifically designed to direct them to positions close to
targets but not exactly towards them unless they are within a certain
\textquotedblleft capture distance,\textquotedblright\ the motivation being
that no agent should be committed to a target until the latest possible time
so as to hedge against the uncertainty of new, potentially more attractive,
randomly appearing targets. This hedging effect is helpful in handling such
uncertainties, but it can create excessive loss of time, especially when
rewards are declining fast. This can be addressed by more direct movements
towards targets, while also re-evaluating the control frequently enough. The
feasible control set in the original CRH is the continuous set $[0,2\pi]^{N}$,
and by appropriately reducing this to a discrete set of control values we will
show how we can eliminate unnecessary hedging. This also reduces the
complexity of the optimal control problem at each time step and facilitates
the problem solution over a finite number of evaluations.

\emph{Estimation of reward-to-go}\textit{:} In the original CRH control
scheme, the visit times are estimated as the earliest time any
agent $j$ would reach some target $i$, given a control $\mathbf{u}_{k}$ at
time $t_{k}$ and maintained over $(t_{k},t_{k}+H_{k}]$. Thus, the estimated
visit time $\tilde{\tau}_{lj}(\mathbf{u}_{k},t_{k},H_{k})$ for any
$l\in\mathcal{T}_{k}$ is
\[
\resizebox{0.47 \textwidth}{!}{$ \tilde{\tau}_{lj}({\bf u}_{k},t_k,H_k)=t_k+H_k+\mathit{d}({\bf x}_j(t_k+H_k,u_j(t_k)),{\bf y}_l)$}\label{tauold}%
\]
where $\mathbf{x}_{j}(t_{k}+H_{k},u_{j}(t_{k}))$ is the location of the agent $j$ in the next time step given the control $u_{j}(t_{k})$. This is a lower bound for $\tilde{\tau}_{ij}$ feasible only when $M_k \le N$, leading also to a mostly unattainable upper bound for
the total reward. We will show how this estimate is improved by a more
accurate projection of each agent's future trajectory.

\section{The New CRH Controller}

\label{NewCRH} In this section we present a new version of the CRH
controller in \cite{Li2006}. Using the definition of $\mathbf{x}_{j}(t_{k}+H_{k},u_{j}(t_{k}))$ given above and assuming $V_{j}=1$ for all agents, the
feasible set for $\mathbf{x}_{j}(t_{k}+H_{k},u_{j}(t_{k}))$ is defined as
\begin{equation}
\mathcal{F}_{j}(t_{k},H_{k})=\{\mathbf{w}\in S|\mbox{   }\mathit{d}%
(\mathbf{w,x}_{j}(t_{k}))=H_{k}\}\label{Fj}%
\end{equation}
In a Euclidean mission space with no obstacles, $\mathcal{F}_{j}(t_{k},H_{k})$ is the circle
centered at $\mathbf{x}_{j}(t_{k})$ with radius $H_{k}$. Let $q_{i}%
(\mathbf{x}_{j}(t))=\mathbbm{1}\{\mathit{d}(\mathbf{x}_{j}(t),\mathbf{y}%
_{i})\leq s_{i}\}$ be the indicator function capturing whether agent $j$
visits target $i$ at time $t$. We define the immediate reward at $t_{k}$:
\begin{equation}
\resizebox{0.47 \textwidth}{!}{$ J_{\bf I}({\bf u}_k,t_k,H_k)=\sum\limits_{j=1}^N \sum\limits_{l=1}^{M_k}{\lambda_{l}\phi_{l}(t_k+H_k)q_{l}({\bf x}_j(t_k+H_k,u_j(t_k)))}$}\label{JI}%
\end{equation}
Following the definition of $\tau_{i}$ as the visit time of target $i$ in
\eqref{tauformulation}, we define $\tilde{\tau}_{ij}$ as the estimated
visit time of target $i$ by agent $j$. Here $\tilde{\tau}_{ij}>t_{k}$ and
any of the agents in the mission space has a chance to visit target $i$. At
time $t_{k}$ we define an estimate of the reward-to-go $J(t_{k+1},H_{k+1})$
for each $\mathbf{u}_{k}$ as
\begin{align}
\tilde{J} &  (\mathbf{u}_{k},t_{k+1},H_{k+1})=\label{Jk}\\
&  \sum_{j=1}^{N}\sum_{l=1}^{M_{k+1}}{\lambda_{l}\phi_{l}(\tilde{\tau}%
_{lj}(\mathbf{u}_{k},t_{k},H_{k}))\cdot q_{l}(\mathbf{x}_{j}(\tilde{\tau}%
_{lj}(\mathbf{u}_{k},t_{k},H_{k})))}\nonumber
\end{align}
We previously mentioned that the original CRH control approach used a lower
bound for estimating $\tilde{\tau}_{ij}$. We improve this estimate and at the
same time address the other two limitations presented above through three
modifications: $(i)$ We introduce a new \emph{travel cost} for
each target, which combines the distance of a target from agents, its
reward, and a local sparsity factor. $(ii)$ We introduce an\emph{ active
target set} associated with each agent at every control evaluation instant
$t_{k}$. This allows us to reduce the infinite dimensional feasible control
set at $t_{k}$ to a finite set. $(iii)$ We introduce a new
event-driven action horizon $h_{k}$ which makes use of the active target set
definition. With these three modifications, we finally present a new CRH
control scheme based on a process of looking ahead over a number of CRH
control steps and aggregating the remainder of a mission through a
reward-to-go estimation process.

\textbf{Travel Cost Factor:} At each control iteration instant $t_{k}$, we
define $\zeta_{i}(t_{k})$ for target $i$ to measure the sparsity of rewards in its vicinity. Let $\bar{D}_{i}>0$ be such that $\phi_{i}(\bar{D}_{i})=0$ for each $i\in\mathcal{T}$ and set $D_{i}=\min(\bar{D}_{i},T)$ so that the
average reward decreasing rate of $i$ over the mission is given by
$\lambda_{i}/D_{i}$. Let the set $\{1,2,...,I\}$ contain the indices of the
$I$ closest targets to $i$ at time $t_{k}$. We then define the \emph{sparsity
factor} for target $i$ as
\begin{equation}
\zeta_{i}(t_{k})=\sum\limits_{l=1}^{I}\gamma^{l}\frac{\mathit{d}%
(\mathbf{y}_{i},\mathbf{y}_{l})}{\lambda_{l}/D_{l}}\label{zeta}%
\end{equation}
where $\gamma\in\lbrack0~1]$ is a parameter used to shift the weight among the
$I$ targets. Note that $\zeta_{i}(t_{k})$ is time-dependent since the set of
$I$ closest targets changes over time as rewards are collected. A larger
$\zeta_{i}(t_{k})$ implies that target $i$ is located in a relatively sparse
area and vice versa. The parameter $I$ is chosen based on the number of
targets in the mission space and the computation capacity of the controller.
The main idea for $\zeta_{i}(t_{k})$ comes from \cite{Schneider2010} where it
was used to solve TSP problems with clustering. Next, for any point in
$\mathbf{x}\in S$, we define target $i$'s travel cost at time $t_{k}$ as
\begin{equation}
\eta_{i}(\mathbf{x},t_{k})=\frac{\mathit{d}(\mathbf{x},\mathbf{y}_{i}%
)}{\lambda_{i}/D_{i}}+\zeta_{i}(t_{k})\label{eta}%
\end{equation}
The travel cost is proportional to the distance metric, so the farther a
target is from $\mathbf{x}$ the more costly is the visit to that target.
It is inversely proportional to the reward's average decreasing rate, implying
that the faster the reward decreases, the less the travel cost is. Adding
$\zeta_{i}(t_{k})$ gives a target in a sparse area a higher travel cost as
opposed to one where there is an opportunity for a visiting agent to collect
additional rewards from its vicinity.

\textbf{Active Targets:} At each control iteration instant $t_{k}$, we define
for each agent $j$ a subset of targets with the following property relative to
the planning horizon $H_{k}$:
\begin{align}
S_{j} &  (t_{k},H_{k})=\big\{\ell|\exists~\mathbf{x}\in\mathcal{F}_{j}%
(t_{k},H_{k})\label{setpi}\\
&  \mbox{ s.t. }\ell=\argmin\limits_{i\in\mathcal{T}_{k}}\eta_{i}%
(\mathbf{x},t_{k}+H_{k}),~i=1,2,...,M_{k}\big\}\nonumber
\end{align}
This is termed the \emph{active target set} and (\ref{setpi}) implies that
$i\in\mathcal{T}_{k}$ is an active target for agent $j$ if and only if it has
the smallest travel cost from at least one point on the reachable set
$\mathcal{F}_{j}(t_{k},H_{k})$. This means that every $\mathbf{x}%
\in\mathcal{F}_{j}(t_{k},H_{k})$ is associated with one of the active targets
and, therefore, so does every feasible heading $u_{j}(t_{k})$. which
corresponds to active target $l$ if and only if:
\begin{equation}
l=\argmin\limits_{i\in\mathcal{T}_{k}}\eta_{i}\big(\mathbf{x}(t_{k}%
+H_{k},u_{j}(t_{k})),t_{k}+H_{k}\big)\label{headingactive}%
\end{equation}
When $\mathit{d}(x,y)$ is continuous, active targets partition the reachable
set $\mathcal{F}_{j}(t_{k},H_{k})$ into several arcs as illustrated in Fig.
\ref{activetarget} where, for simplicity, we assume $\gamma=0$ in \eqref{zeta}
and all $\lambda_{i}$ and $\phi_{i}(t)$, $i=1,\ldots,M$ are the same. In this
case, agent 1 has four active targets: $S_{1}(t_{k},H_{k})=\{1,2,4,5\}$. The
common feature of all points on an arc is that they correspond to the same
active target with the least travel cost. 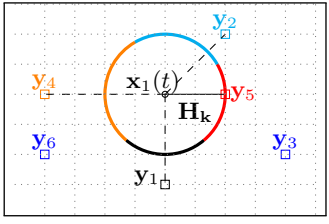
\begin{figure}[ptb]
\centering
\setlength\fboxsep{0pt} \setlength\fboxrule{0.5pt} \fbox{
\begin{tikzpicture}[scale=0.8]
\draw(5,4) circle (1);
\color{black}
\draw(5,4) circle (0.05);
\draw(4.8,4.2) node{\small ${\bf x}_1(t)$};
\color{black}
\draw(5,2.5) \Square{2pt} ;
\draw (4.7,2.6) node{\small ${\bf y}_1$};
\color{cyan}
\draw(6,5) \Square{2pt};
\draw (6,5.2) node{\small ${\bf y}_2$};
\color{blue}
\draw(7,3) \Square{2pt};
\draw (7,3.2) node{\small ${\bf y}_3$};
\draw(3,3) \Square{2pt};
\draw (3,3.2) node{\small ${\bf y}_6$};
\color{orange}
\draw(3,4) \Square{2pt};
\draw (3,4.2) node{\small ${\bf y}_4$};
\color{red}
\draw (6,4) \Square{2pt};
\draw (6.3,4) node{\small ${\bf y}_5$};;
\color{black}
\draw (5,4)--(6,4);
\draw (5.5,3.7) node{\small ${\bf H_k}$};
\draw [red,very thick] (5.6428,3.2340) arc [radius=1, start angle=-50, end angle= 30];
\draw [cyan,very thick] (5.866,4.5) arc [radius=1, start angle=30, end angle= 120];
\draw [orange,very thick] (4.5,4.866) arc [radius=1, start angle=120, end angle= 230];
\draw [black,very thick] (5.6428,3.2340) arc [radius=1, start angle=-50, end angle= -130];
\draw [dashed] (5,4)--(5,2.5);
\draw [dashed] (5,4)--(6,5);
\draw [dashed] (5,4)--(3,4);
\draw[thin,gray,dotted,step=.5] (2.5,2) grid (7.5,5.5);
\end{tikzpicture}
}\caption{The Active Target Set for agent 1: $S_{1}(\mathbf{x}_{1}%
(t_{k}),H_{k})=\{1,2,4,5\}$}%
\label{activetarget}%
\end{figure}

\textit{Construction of }$S_{j}(t_{k},H_{k})$\textit{:} For each target
$l\in\mathcal{T}_{k}$ and each agent $j$, let $\mathcal{L}_{k}(\mathbf{x}%
_{j}(t_{k}),l)$ be the set of points $\mathbf{x}\in S$ defining the shortest
path from $\mathbf{x}_{j}(t_{k})$ to $\mathbf{y}_{l}$. The intersection of
this set with $\mathcal{F}_{j}(t_{k},H_{k})$ is the set of closest points to
target $l$ in the feasible set:
\begin{equation}
\mathcal{C}_{l,j}(t_{k},H_{k})=\mathcal{L}_{k}(\mathbf{x}_{j}(t_{k}%
),l)\cap\mathcal{F}_{j}(t_{k},H_{k})\label{cp}%
\end{equation}
In a Euclidean mission space, $\mathcal{L}_{k}(\mathbf{x}_{j}(t_{k}),l)$ is a
convex combination (line segment) of $\mathbf{x}_{j}(t_{k})$ and
$\mathbf{y}_{l}$, while $\mathcal{C}_{l,j}(t_{k},H_{k})$ is a single point
where this line crosses the circle $\mathcal{F}_{j}(t_{k},H_{k})$. The
following lemma provides a necessary and sufficient condition for identifying
targets which are active for an agent at $t_{k}$ using $\mathcal{C}_{l,j}(t_{k},H_{k})$.

\begin{lemma}
\label{MRCPlem1} Target $l$ is an active target for agent $j$ at time $t_{k}$
if and only if, $\forall i\in\mathcal{T}_{k}$
\begin{equation}
\eta_{l}(\mathcal{C}_{l,j}(t_{k},H_{k}),t_{k+1})\leq\eta_{i}(\mathcal{C}%
_{l,j}(t_{k},H_{k}),t_{k+1}) \label{MRCPlemma1statement}%
\end{equation}
\end{lemma}
\begin{proof}
See Appendix.
\end{proof}

\textbf{Action Horizon:} The definition of $h_{k}$ in \cite{Li2006} requires
frequent iterations of the optimization problem through which $\mathbf{u}^*_{k}$
is determined in case no random event is observed to justify such action.
Instead, when there are no random events, we define a new \emph{multiple
immediate target event} to occur when the minimization in \eqref{HkDef}
returns more than one target meaning the agent is at an equal distance to at least two targets. We then define $h_{k}$ to be the shortest time until the first multiple immediate target event occurs in $(t_{k},t_{k}+H_{k}]$:
\begin{align}
\label{hk}
h_k&=\min\Big\{H_k,\inf\big\{t>t_k:\exists l,l^*\in\mathcal{T}_k \mbox{ s.t.} \\\nonumber&\mathit{d}(\mathbf{x}_j(t_k+t,u_{j}(t_{k})),\mathbf{ y}_l)=\mathit{d}(\mathbf{x}_{j}(t_{k}+t,u_{j}(t_{k})),\mathbf{y}_{l^*})   \big\}\Big\}
\end{align}
%By (\ref{headingactive}), the change in $S_{j}(t_{k},H_{k})$ is in
%fact the earliest time when a CRH control re-evaluation is needed (unless an
%uncontrollable random event occurs) since the feasible control set remains
%otherwise unaffected.
Consequently, this definition of $h_{k}$ eliminates any unnecessary control evaluation.

\subsection{Look Ahead and Aggregate Process}

In order to solve the optimization problem $\mathbf{P}_{k}$ in
\eqref{dp_problem} using the CRH approach, we need the estimated visit
time $\tilde{\tau}_{ij}(\mathbf{u}_{k},t_{k},H_{k})$ for each $\mathbf{u}_{k}$
through which $\tilde{J}(\mathbf{u}_{k},t_{k+1},H_{k+1})$ in (\ref{Jk}) can be
evaluated. This estimate is obtained by using a projected path for each agent.
This path projection consists of a \emph{look ahead} and an \emph{aggregate}
step. In the first step, the active target set $S_{j}(t_{k},H_{k})$ is
determined for agent $j$. With multiple agents in a mission, at each iteration
step the remaining targets are partitioned using the relative proximity
function in \eqref{proximity}. We denote the target subset for agent $j$ as
$\mathcal{T}_{k,j}$ where:
\begin{equation}
l\in\mathcal{T}_{k,j}\mbox{  }\Longleftrightarrow\mbox{  }p(\delta_{lj}%
(t_{k}))>p(\delta_{lq}(t_{k}))\mbox{  }\forall q\in\mathcal{A}%
\end{equation}
Let $|\mathcal{T}_{k,j}|=M_{k,j}$. All $\tilde{\tau}_{ij}(\mathbf{u}_{k}%
,t_{k},H_{k})$ are estimated as if $j$ would visit targets in its own subset
by visiting the one with the least travel cost first.
%If target $l \not\in \mathcal{T}_{k,j}$ then $\tilde\tau_{ij}({\bf u}_k,t_k,H_k)=0$ to ensure agent $j$ gains no reward from target $l$.
We define the agent $j$'s tour as the permutation ${\boldsymbol{\theta}}%
^{j}(\mathbf{u}_{k},t_{k},H_{k})$ specifying the order in which it visits
targets in $\mathcal{T}_{k,j}$. For simplicity, we write ${\boldsymbol{\theta
}}^{j}$ and let ${\boldsymbol{\theta}^{j}_i}$ denote the $i^{th}$ target in
agent $j$'s tour. Then, for all $l\in\mathcal{T}_{k,j}$ and $t_{k+1}%
=t_{k}+H_{k}$:
\[
\eta_{{\boldsymbol{\theta}}^{j}_1}(\mathbf{x}_{j}(t_{k+1},u_{j}%
(t_{k})),t_{k+1})\leq\eta_{l}(\mathbf{x}_{j}(t_{k+1},u_{j}(t_{k})),t_{k+1})
\]
and with $n=2,...,M_{k,j}-1$, for all $l\in\mathcal{T}_{k,j}-\{\theta
^{j}_1,...,\theta^{j}_n\}$,
\[
\eta_{{\boldsymbol{\theta}}^{j}_{n+1}}(\mathbf{y}_{{\boldsymbol{\theta}}%
^{j}_n},\tilde{\tau}_{{\boldsymbol{\theta}}^{j}_n}(\mathbf{u}_{k},t_{k},H_{k}))\leq\eta_{l}(\mathbf{y}_{{\boldsymbol{\theta}}^{j}_n}%
,\tilde{\tau}_{{\boldsymbol{\theta}}^{j}_n}(\mathbf{u}_{k},t_{k},H_{k}))
\]
where
\begin{equation}
\tilde{\tau}_{{\boldsymbol{\theta}}^{j}_n}(\mathbf{u}_{k},t_{k},H_{k}%
)=t_{k}+H_{k}+\sum\limits_{i=1}^{n-1}\mathit{d}(\mathbf{y}%
_{{\boldsymbol{\theta}}^{j}_i},\mathbf{y}_{{\boldsymbol{\theta}}^j_{i+1}})\label{tautilde}%
\end{equation}
This results in the corresponding $\tilde{\tau}_{lj}(\mathbf{u}_{k}%
,t_{k},H_{k})$ for all $l\in\mathcal{T}_{k,j}$. We can now obtain the
reward-to-go estimate as
\begin{align}
J_{\mathbf{A}} &  (\mathbf{u}_{k},t_{k},H_{k})=\label{partitionJK}\\
&  \sum_{j=1}^{N}\sum_{l=1}^{M_{k+1,j}}{\lambda_{l}\phi_{l}(\tilde{\tau}%
_{lj}(\mathbf{u}_{k},t_{k},H_{k}))\cdot q_{l}(\mathbf{x}_{j}(\tilde{\tau}%
_{lj}(\mathbf{u}_{k},t_{k},H_{k})))}\nonumber
\end{align}
%The total reward is then calculated as:
%\begin{equation*}
%J({\bf u}_k,t_k,H_k)=J_{\bf I}({\bf u}_k,t_k,H_k)+J_{\bf A}({\bf u}_k,t_k,H_k)
%\end{equation*}
Recalling the immediate reward in \eqref{JI}, the optimization problem
$\mathbf{P}_{k}$ becomes:
\begin{equation}
\max\limits_{\mathbf{u}_{k}\in\lbrack0~2\pi]^{N}}\big[J_{\mathbf{I}%
}(\mathbf{u}_{k},t_{k},H_{k})+J_{\mathbf{A}}(\mathbf{u}_{k},t_{k}%
,H_{k})\big]\label{Pkest}%
\end{equation}
In \eqref{Fj} we defined the feasible set for the location of agent $j$ in the
next step $t_{k+1}=t_{k}+H_{k}$. In a Euclidean mission space, each point
$\mathbf{x}\in\mathcal{F}_{j}(t_{k},H_{k})$ corresponds to a heading
$v(\mathbf{x})$ relative to the agent's location $x_{j}(t_{k})$. Using the definition in \ref{cp} let:
\[
\mathcal{V}_{j}(t_{k},H_{k})=\big\{v(\mathbf{x})|\mathbf{x}=\mathcal{C}%
_{l,j}(t_{k},H_{k}),\text{ \ }l\in S_{j}(t_{k},H_{k})\big\}\label{Vj}%
\]
and
\[
\mathbf{V}_{k}=\mathcal{V}_{1}(t_{k},H_{k})\times\mathcal{V}_{2}(t_{k}%
,H_{k})\times...\times\mathcal{V}_{N}(t_{k},H_{k})
\]  In the next lemma, we prove that in a single-agent mission with the objective function
defined in \eqref{Pkest} the optimal control is $u_{1}(t_{k})=v(\mathcal{C}%
_{l,1}(t_{k},H_{k}))$ for some $l\in S_{1}(t_{k},H_{k})$.
\begin{lemma}
\label{MRCPlem2} In a single agent $(N=1)$ mission, if $u_{1}^{\ast}$ is an
optimal solution to the problem:
\begin{equation}
\max\limits_{u_{k}\in\lbrack0~2\pi]}\big[J_{\mathbf{I}%
}(u_{k},t_{k},H_{k})+J_{\mathbf{A}}(u_{k},t_{k},H_{k})\big]
\end{equation}
then $u_{1}^{\ast} \in \mathcal{V}_{j}(t_{k},H_{k})$
\end{lemma}
\begin{proof}
See Appendix.
\end{proof}
The implication of this lemma is that we can reduce the number of
feasible controls to a finite set as opposed to the infinite set $[0,2\pi]$.
\begin{theorem}
\label{MRCPtheorem1} In a multi-agent MRCP mission, if $\mathbf{u}^{\ast
}=[u_{1}^{\ast},...,u_{N}^{\ast}]$ is the optimal solution to the problem in
\eqref{Pkest} then $\mathbf{u}^{\ast}\in\mathbf{V}_{k}$.
\end{theorem}
\begin{proof}
See Appendix
\end{proof}
Theorem \ref{MRCPtheorem1} reduces the problem $\mathbf{P}_{k}$ to
a maximization problem over a finite set of feasible controls:
\[
\max\limits_{\mathbf{u}_{k}\in\mathbf{V}_{k}%
}\big[J_{\mathbf{I}}(\mathbf{u}_{k},t_{k},H_{k})+J_{\mathbf{A}}(\mathbf{u}%
_{k},t_{k},H_{k})\big]\label{CountableJ}%
\]
This reduces the size of the problem compared to the original CRH controller.
The following algorithm generates controls in this manner at each step $t_{k}$
and is referred to as the \textquotedblleft One-step
Lookahead\textquotedblright\ CRH controller (extended to a \textquotedblleft%
$K$-step Lookahead\textquotedblright\ algorithm in what follows).\\

{\bf CRH One-step Lookahead Algorithm:}
{\small
\begin{enumerate}
\item Determine $H_{k}$ through (\ref{HkDef}).
\item Determine the active target set $S_{j}(t_{k},H_{k})$ through
(\ref{setpi}) for all $j\in\mathcal{A}$.
\item Evaluate $J_{\mathbf{A}}(\mathbf{u}_{k},t_{k},H_{k})$ for all
$\mathbf{u}_{k}\in\mathbf{V}_{k}$ through \eqref{tautilde} and \eqref{partitionJK}
\item Solve $\mathbf{P}_{k}$ in (\ref{Pkest}) and determine $\mathbf{u}%
_{k}^{\ast}$.
\item Evaluate $h_{k}$ through \eqref{hk}
\item Execute $\mathbf{u}_{k}^{\ast}$ over $(t_{k},t_{k}+h_{k}]$ and repeat
Step 1 with $t_{k+1}=t_{k}+h_{k}$.
\end{enumerate}}
\subsection{$K$-Step Lookahead}

The One-step Lookahead CRH controller can be extended to a $K$-step Lookahead
controller with $K>1$ by exploring additional possible future paths for each
agent at each time step $t_{k}$. In the One-step Lookahead algorithm, the optimal reward-to-go is estimated based on a single tour over the remaining targets.
The $K$-step Lookahead algorithm estimates this reward by considering more
possible tours for each agent as follows. For any feasible $u_{j}(t_{k}%
)\in\mathcal{V}_{j}(t_{k},H_{k})$ the agent is hypothetically placed at the
corresponding next step location $\mathbf{x}_{j}(t_{k+1})$. This is done for
all agents to maintain synchronicity of the solution. At $\mathbf{x}%
_{j}(t_{k+1})$, a new active target set is determined, implying that
agent $j$ can have $|S_{j}(t_{k}+H_{k},H_{k+1})|$ possible paths. At this
point, we can repeat the same procedure by hypothetically moving the agent to
a new feasible location from the set $\mathcal{F}_{j}(t_{k+1},H_{k+1})$ or we
can stop and estimate the reward-to-go for each available path. Thus, for a
Two-Step Lookahead, problem $\mathbf{P}_{k}$ becomes:
\begin{align}
\nonumber
\max_{\mathbf{u}_{k}\in\mathbf{V}_{k}} &  \Big[J_{\mathbf{I}}(\mathbf{u}%
_{k},t_{k},H_{k})+\max_{\mathbf{u}_{k+1}\in\mathbf{V}_{k+1}}\big[J_{\mathbf{I}%
}(\mathbf{u}_{k+1},t_{k+1},H_{k+1})\\\label{2Step}
&  +J_{\mathbf{A}}(\mathbf{u}_{k+1},t_{k+1},H_{k+1})\big]\Big]
\end{align}
We extend the previous algorithm to a 2-step lookahead in the following. For a $K$-step we should repeat steps 1 and 2 for $K$ times before moving to step 4 of the algorithm.\\

\textbf{CRH 2-step Lookahead Algorithm:}
{\small
\begin{enumerate}
\item Determine $H_{k}$ through (\ref{HkDef}).
\item Determine the active target set $S_{j}(t_{k},H_{k})$ through
(\ref{setpi}) for all $j\in\mathcal{A}$.
\item Repeat steps $1\&2$ for $t_{k+1}=t_k+H_k$ and ${\mathbf x}_j(t_{k+1})={\mathbf x}_j(t_{k+1},u_j(t_k))$ for all $u_j(t_k)\in\mathcal{V}_j(t_k,H_k)$.
\item Evaluate $J_{\mathbf{A}}(\mathbf{u}_{k+1},t_{k+1},H_{k+1})$ for all
$\mathbf{u}_{k+1}\in\mathbf{V}_{k+1}$ through \eqref{tautilde} and \eqref{partitionJK}
\item Solve $\mathbf{P}_{k}$ in (\ref{2Step}) and determine $\mathbf{u}%
_{k}^{\ast}$.
\item Evaluate $h_{k}$ through \eqref{hk}
\item Execute $\mathbf{u}_{k}^{\ast}$ over $(t_{k},t_{k}+h_{k}]$ and repeat
Step 1 with $t_{k+1}=t_{k}+h_{k}$.
\end{enumerate}}
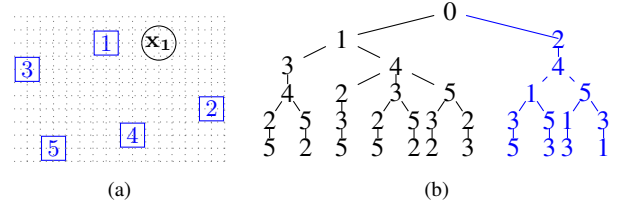
\begin{figure}
\centering
\begin{subfigure}[]
{\begin{tikzpicture}[scale=0.35]
\color{black}
\draw(7,5) circle (0.65);
\draw(7,5) node{\small ${\bf x_1}$};
\color{blue}
\draw(5,5) \Square{13pt};
\draw (5,5) node{\small $1$};
\draw(9,2.5) \Square{13pt};
\draw (9,2.5) node{\small $2$};
\draw(2,4) \Square{13pt};
\draw (2,4) node{\small $3$};
\draw(6,1.5) \Square{13pt};
\draw (6,1.5) node{\small $4$};
\draw(3,1) \Square{13pt};
\draw (3,1) node{\small $5$};
\draw[thin,gray,dotted,step=.5] (1.5,0.5) grid (9.5,6);
\end{tikzpicture}
\label{5targetmission}}
\end{subfigure}
\begin{subfigure}[]
{\begin{tikzpicture}[scale=0.24,every node/.style={draw=none},level 1/.style={sibling distance=120mm},level 2/.style={sibling distance=60mm}, level 3/.style={sibling distance=30mm}, level 4/.style={sibling distance=20mm}]
\node (root) {0}
child{ node{1}
child{ node{3}
child{node{4}
child {node {2}
child {node {5}}}
child {node {5}
child {node {2}}}}}
child{ node{4}
child{ node{2}
child{ node{3}
child{ node{5}}}}
child{ node{3}
child {node {2}
child {node {5}}}
child {node {5}
child {node {2}}}}
child{ node{5}
child{ node{3}
child{ node{2}}}
child{ node{2}
child{ node{3}}}}}}
child{[blue] node{2}
child{  node{4}
child{ node{1}
child {node {3}
child {node {5}}}
child {node {5}
child {node {3}}}}
child{ node{5}
child {node {1}
child {node {3}}}
child{ node{3}
child{ node{1}}}}}}
;
\end{tikzpicture}
\label{5targetmissiontree}}
\end{subfigure}
\caption{(a): Five-target mission, (b): The tree structure }%
\end{figure}
This procedure can easily be repeated and the whole process can be
represented as a tree structure where the root is the initial location of the
agent and a path from the root to each leaf is a possible target sequence for
the agent. In Fig. \ref{5targetmission} a sample mission with 5 targets is
shown with its corresponding tree in Fig. \ref{5targetmissiontree}. A
brute-force method involves $5!=120$ possible paths, whereas the tree
structure for this mission is limited to 11 paths. The active target set for
agent 1 consists of targets ${1,2}$. Each of these active targets would then
generate several branches in the tree, as shown. We calculate the total reward
for each path to find the optimal one. Determining the complete tree for large
$K$ is time consuming. The $K$-step Lookahead CRH controller enables us to
investigate the tree down to a few levels and then calculate an estimated
reward-to-go for the rest of the selected path. However, there is no guarantee on the monotonicity of the results with more lookahead steps and in some cases the final result degrades with one more lookahead step.
\subsection{Two-Target, One-Agent Case}
The simplest case of the MRCP is the case with one agent and two targets.
Obviously, this is an easy routing problem whose solution is one of the two
possible paths the agent can take. We prove that the One-step Lookahead
algorithm solves the problem with any linearly decreasing reward function.
Consider a mission with one agent and two targets with initial rewards and
deadlines $\lambda_{1},D_{1}$ and $\lambda_{2},D_{2}$ respectively. The
analytical solution for this case reveals whether path $\theta_{1}=(1,2)$ or
$\theta_{2}=(2,1)$ is optimal. Following the previous analysis, we assume that
$V_{1}=1$ and set $\mathbf{x}_{1}(t_{k})=\mathbf{x}$ for the sake of brevity.
We also assume the rewards are linearly decreasing to zero: $\phi
_{i}(t)=1-\frac{t}{D_{i}}$. The two possible rewards are given by:
\begin{equation}
R_{(1,2)}=\lambda_{1}\big[1-\frac{\mathit{d}(\mathbf{x},\mathbf{y}_{1})}%
{D_{1}}]+\lambda_{2}\big[1-\frac{\mathit{d}(\mathbf{x},\mathbf{y}%
_{1})+\mathit{d}(\mathbf{y}_{1}-\mathbf{y}_{2})}{D_{2}}]\label{R12}%
\end{equation}%
\begin{equation}
R_{(2,1)}=\lambda_{2}\big[1-\frac{\mathit{d}(\mathbf{x},\mathbf{y}_{2})}%
{D_{2}}]+\lambda_{1}\big[1-\frac{\mathit{d}(\mathbf{x},\mathbf{y}%
_{2})+\mathit{d}(\mathbf{y}_{2}-\mathbf{y}_{1})}{D_{1}}]\label{R21}%
\end{equation}
Therefore, if $R_{(1,2)}>R_{(2,1)}$, it follows that the following inequality
must hold:
\begin{gather}
\frac{\lambda_{1}}{D_{1}}\big[\mathit{d}(\mathbf{x},\mathbf{y}_{1}%
)-\mathit{d}(\mathbf{x},\mathbf{y}_{2})+\mathit{d}(\mathbf{y}_{2}%
,\mathbf{y}_{1})\big]<\nonumber\\
\quad\frac{\lambda_{2}}{D_{2}}\big[\mathit{d}(\mathbf{x},\mathbf{y}%
_{2})-\mathit{d}(\mathbf{x},\mathbf{y}_{1})+\mathit{d}(\mathbf{y}%
_{1},\mathbf{y}_{2})\big]\label{Twotargetanalytic}%
\end{gather}
and the optimal path is $\theta^{\ast}=\theta_{1}$
%\begin{lemma}
%Consider a mission with two target and one agent where the agent is located at ${\bf x}_1(t)$ and target 1 and 2 are located at ${\bf y}_1$ and ${\bf y}_2$. Target i's reward at time $t$ is $\lambda_i(1-\frac{t}{D_i})$ and WLOG $\mathit{d}({\bf x},{\bf y}_1)<\mathit{d}({\bf x},{\bf y}_2)$. \\
%Assuming $\gamma=0$ in \eqref{zeta}, if target 2 is not an active target at time $t$ then the path ${\bf x}_1(t)\rightarrow {\bf y}_1\rightarrow {\bf y}_2$ is optimal and the CRH controller finds the optimal path.
%\begin{proof}
%See the Appendix.
%\end{proof}
%\label{MRCPlem3}
%\end{lemma}
. Letting $\theta^{CRH}$ denote the path obtained by the One-step Lookahead
CRH controller, we show next that this controller recovers the optimal path
$\theta^{\ast}$.

\begin{theorem}
\label{twotargettheorem} Consider a two-target, one-agent mission. If $\gamma=0$ in \eqref{zeta} and target $i$'s reward at
time $t$ is $\lambda_{i}(1-\frac{t}{D_{i}})$, then $\theta^{CRH}=\theta^{\ast
}$.\end{theorem}
\begin{proof}
See Appendix.
\end{proof}

\subsection{Monotonicity in the Look Ahead Steps}
Questions that come into mind after introducing the multiple look ahead steps CRH controller are: How many look ahead steps should we perform? Is it always better to do more look ahead steps? Or in a simple way, does the more steps look ahead always gives a better answer than less?

The answer to the first question is that it depends on the size of the problem and our computation capability. We can even adjust the number of look ahead steps during the course of the solution. We can start with more when there is more targets available and lower the number once there is only a few targets in the mission space.
The answer to the other two questions is No. As much as one would like to have a sort of monotonicity effect in this problem, the complexity of the problem and its significant dependence on the mission topology causes the non-monotone results with different number of look ahead steps. Here we are going to show a case with 10 equally important targets and one agent. This is a straight forward TSP for which the optimal path can be obtained through an exhaustive search. For this case the one and two look ahead steps CRH controllers find the same path with a reward of 92.6683. However, once we move up to three look ahead steps, the CRH controller degrade to a worst path with 92.5253 reward. The path for these controllers is shown in figures \ref{1s} and \ref{3s}. The optimal path that is calculated through the exhaustive search is obtained by the CRH controllers when we go up to six look ahead steps (Fig. \ref{6s}). The observation is that the non-monotone results from higher number of look ahead is a local effect and once we increase the look ahead steps CRH controller can solve the problem to the optimality. This obviously is not the case for all missions and in some cases the optimal path can not be retrieved by CRH controller with any look ahead steps.
\begin{figure}[ht]
\begin{centering}
\begin{subfigure}[One Step Look Ahead: {$[1-9-7-4-3-10-2-6-5-8]$} - Reward=92.6683 - Time=868]
{\includegraphics[scale=0.3]{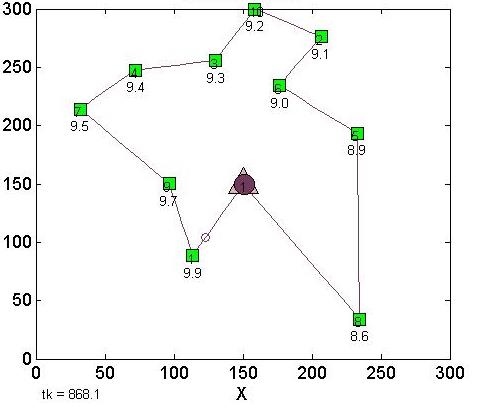}
\label{1s}}
\end{subfigure}
\begin{subfigure}[Three Step Look Ahead: {$[6-2-10-3-4-7-9-1-8-5]$} - Reward=92.5253 - Time=897]
{\includegraphics[scale=0.3]{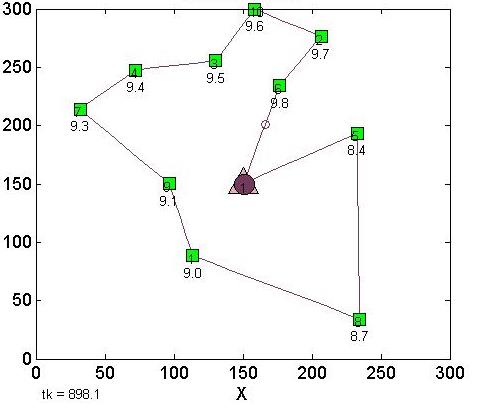}
\label{3s}}
\end{subfigure}\\
\begin{subfigure}[Five Step Look Ahead: {$[5-6-2-10-3-4-7-9-1-8]$} - Reward=92.6031 - Time=862]
{\includegraphics[scale=0.31]{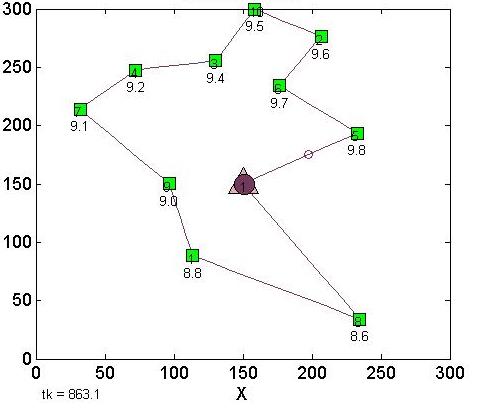}
\label{5s}}
\end{subfigure}
\begin{subfigure}[Six Step Look Ahead: {$[9-7-4-3-10-2-6-5-1-8]$} - Reward=92.7436 - Time=916]
{\includegraphics[scale=0.31]{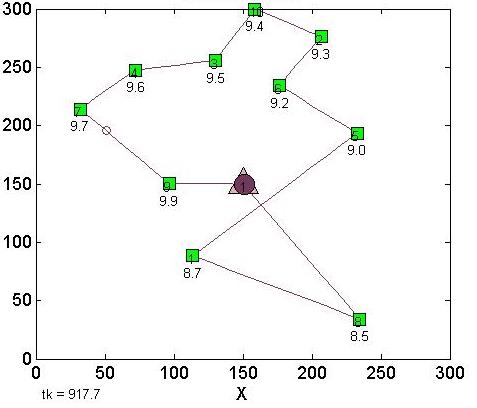}
\label{6s}}
\end{subfigure}
\end{centering}
\label{monotonicity}
\caption{10 Target mission with different number of look ahead steps}
\end{figure}
\section{Simulation Examples}

\label{Numerical} We provide several MRCP examples in which the performance of
the original and new CRH controllers is compared. In all examples, we use
parameters $\Delta=0$, $V_{j}=1$, $\alpha_{i}=1$, $\beta_i=1$.

{\bf TSP Benchmark Comparison:}
We use the CRH controller as a path planning algorithm for
some benchmark TSP problems. Table \ref{TSPtable} shows the result of the
2-step and 3-step Lookahead algorithm compared to the optimal results from
\cite{Reinelt1991}. We emphasize that the CRH controller is not designed for
deterministic TSP problems so it is not expected to perform as well as highly
efficient TSP algorithms. Nonetheless, as a starting basis of comparison, we
note that the errors are relatively small, ranging from 7.8 to 23.8\%.
%\href{http://comopt.ifi.uni-heidelberg.de/software/TSPLIB95/}{TSPLIB} .
\begin{table}[ptb]
\caption{TSP benchmark instances comparison with the CRH controller algorithm}%
\vspace{1mm} \centering
\resizebox{0.85\columnwidth}{!}{
\begin{tabular}{|p{1cm}|p{1cm}|p{1cm}|p{1cm}|p{1cm}|p{1cm}|}
\hline
TSP \quad Instance & Optimal Tour Length & Two Step Lookahead & Three Step Lookahead & Limited Range Agent & Minimum Error ($\%$) \\\hline
att48 & 33522 & 38011 & 37492 & 41112 & 11.8 \\\hline
eil51 & 426 & 547 & 480 & 507 & 12.6 \\\hline
berlin52 & 7542 & 8713 & 8713 & 8137 & 7.8 \\\hline
st70 & 675 & 840 & 818 & 816 & 20.8 \\\hline
eil76 & 538 & 633 & 635 & 655 & 17.6 \\\hline
pr76 & 108159 & 146980 & 131678 & 146944 & 21.7 \\\hline
rat99 & 1211 & 1451 & 1470 & 1591 & 19.8 \\\hline
rd100 & 7910 & 9529 & 9123 & 9618 & 15.3 \\\hline
kroA100 & 21282 & 25871 & 24795 & 23782 & 11.7 \\\hline
kroB100 & 22141 & 28093 & 27415 & 28581 & 23.8 \\\hline
kroC100 & 20749 & 24603 & 25561 & 26171 & 18.5 \\\hline
\end{tabular}
}\label{TSPtable}
%\vspace{-12mm}
\end{table}

In an attempt to measure the sensitivity of the results of the new CRH
controller to partial mission information, we also tested cases where agents
have limited sensing range (see fifth column in table \ref{TSPtable}). In
these cases, the agent only senses a target if it is within its sensing range
which we have assumed to be $20\%$ of the maximum dimension of the mission
space. The results in most cases are comparable to the full-information cases.
The computation time for the limited range agents is about an order of
magnitude shorter than the other one. These results show the low sensitivity
of the CRH controller performance to non-local information for each agent.
This observation suggests that CRH controller is likely to provide good
performance in a distributed implementation or in cases where targets are not
known a priori and should be locally sensed by the agents.

{\bf Addressing Instabilities}: As already mentioned, the original CRH controller may give rise to
oscillatory trajectories and fail to complete a mission. This is illustrated
in Fig. \ref{OldCRH3target} for a simple mission with three linearly
discounted reward targets. In Fig. \ref{NewCRH3target}, it is shown that the
new CRH controller can easily determine the optimal path in this simple
case. \begin{figure}[th]
\centering
\begin{subfigure}[Original CRH Oscillation]{
\includegraphics[width=1.5in]{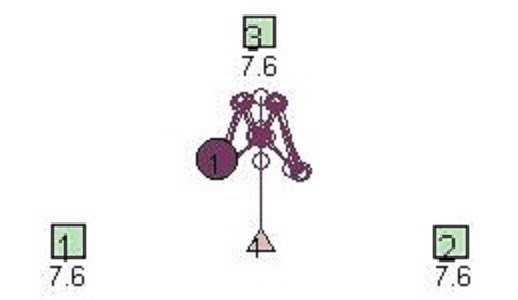}
\label{OldCRH3target}
}
\end{subfigure}
\begin{subfigure}[New CRH Optimal Solution]{
\includegraphics[width=1.5in]{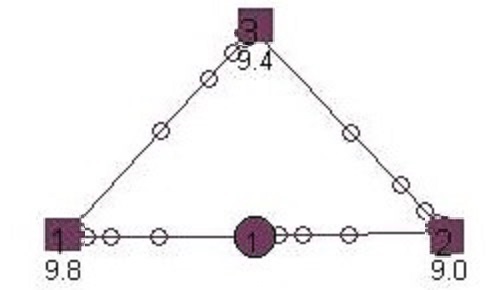}
\label{NewCRH3target}
}
\end{subfigure}
\caption{Comparison of the two CRH controllers for a 3 targets mission}%
\label{3targetcompare}%
\end{figure}

{\bf Comparison between original and new CRH Controller}:
A mission with 25 targets distributed uniformly and 2 agents starting at a
base is considered as shown in  Fig. \ref{mission25}, with uniformly
distributed initial rewards: $\lambda_{i}\sim
U(10,20)$ and $D_{i}\sim U(300,600)$ as in \eqref{reward}. In this case, the original CRH (Fig.
\ref{OldCRH25}) underperforms compared to 3-step and 5-step Lookahead
CRH controller (Figs. \ref{3StepCRH25}, \ref{5StepCRH25}) by a large
margin. We have used a value of $\gamma=0.3$ and $I=25$ in \eqref{zeta}. This
comes at the price of a slightly longer mission time in the 3-Step look ahead case, since the original
controller never reaches some targets before their rewards are lost. However,
minimizing time is not an objective of the MRCP considered here and reward
maximization dictates the final length of the mission.
\begin{figure}[th]
\centering
\begin{subfigure}[Complete Mission]
{ \includegraphics[width=1.5in]{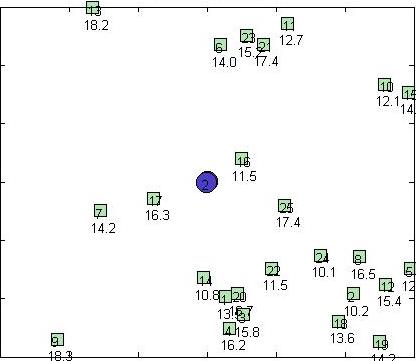}\label{mission25}}
\end{subfigure}
\begin{subfigure}[Original CRH, Reward=62.8, Time=714]
{ \includegraphics[width=1.5in]{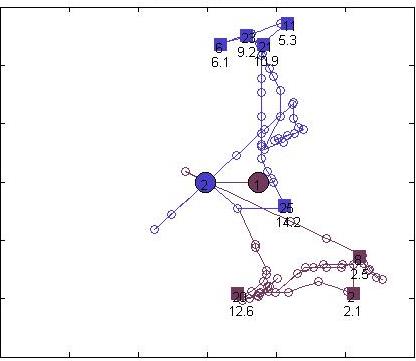}\label{OldCRH25}}
\end{subfigure}
\begin{subfigure}[3-Step Lookahead, Reward=141.29, Time=753]
{ \includegraphics[width=1.51in]{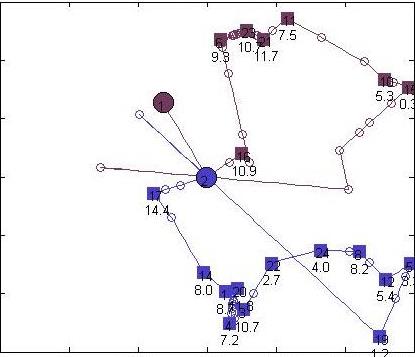}\label{3StepCRH25}}
\end{subfigure}
\begin{subfigure}[5-Step Lookahead, Reward=143.42, Time=657]
{ \includegraphics[width=1.51in]{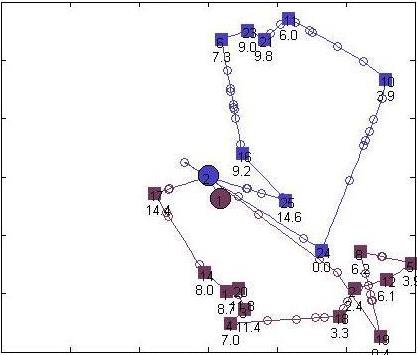}\label{5StepCRH25}}
\end{subfigure}
\caption{Performance comparison of the original and new CRH algorithms}%
\label{comparison2}%
\end{figure}

{\bf Randomly Generated Missions}: To compare the overall performance of the new CRH controller, we
generated 10 missions, each with 20 targets that are uniformly located
in a $300\times300$ mission space and two agents initially at the base. We
have used $\lambda_{i}\sim U(2,12)$ and $D_{i}=300$. The results are shown in
Table \ref{20targets} where we can see that the average total reward is
increased by $22\%$ while the average mission time is increased by $8\%$.
\begin{table}
\caption{20 Target-2 Agent Missions}%

\centering
\resizebox{0.8\columnwidth}{!}{
\begin{tabular}{|p{1.1cm}|p{1cm}|p{1.4cm}|p{1cm}|p{1.4cm}|}
\hline
\multirow{2}{*}{Mission $\#$} & \multicolumn{2}{|c|}{Original CRH} & \multicolumn{2}{|c|}{Three Step Lookahead CRH}\\\cline{2-5}
&Reward & Travel Time & Reward & Travel Time\\
\hline
1&33.92&412&45.24&536\\\hline
2&41.48&439&52.4&426\\\hline
3&30.93&476&41.19&483 \\\hline
4&32.08&389&37.24&457 \\\hline
5&41.5&444&47.25&537 \\\hline
6&44.61&389&47.91&471 \\\hline
7&23.93&528&35.48&462 \\\hline
8&38.68&415&50.91&489 \\\hline
9&30.92&478&34.08&429 \\\hline
10&36.81&458&44.26&476 \\\hline
\bf Average&35.48&443&43.53&479\\\hline
\end{tabular}}\label{20targets}
\end{table}
In another case 10 missions were generated, each with 20 targets where 10 targets are only initially available to the agents. The other 10 targets would randomly appear during the mission. We use an initial reward $\lambda_{i}\sim U(2,12)$ and the parameter $D_{i}\sim U(300,600)$. The comparison of the original and new CRH controller is shown in table \ref{randomtargettable}. An increase of $31\%$ is seen in the total reward with a slight $2\%$ increase in the total mission time.
\begin{table}[H]

\caption{20 Target-2 Agent Missions, With random target appearance}%
\centering
\resizebox{0.8\columnwidth}{!}{
\begin{tabular}{|p{1.1cm}|p{1cm}|p{1.4cm}|p{1cm}|p{1.4cm}|}
\hline
\multirow{2}{*}{Mission $\#$} & \multicolumn{2}{|c|}{Original CRH} & \multicolumn{2}{|c|}{Two Step Lookahead CRH}\\\cline{2-5}
&Reward & Travel Time & Reward & Travel Time\\
\hline
1&51.63&704&58.61&736\\\hline
2&46.53&716&67.57&632\\\hline
3&37.59&646&54.42&691\\\hline
4&31.71&929&53.13&941\\\hline
5&60.05&668&81.31&528\\\hline
6&58.16&609&67.91&688\\\hline
7&45.81&739&61.29&760\\\hline
8&49.71&722&59.47&732\\\hline
9&42.64&822&47.95&818\\\hline
10&40.32&648&58.01&868\\\hline
\bf Average&46.42&720&60.97&739\\\hline
\end{tabular}}\label{randomtargettable}\end{table}

{\bf Sparsity Factor in Clustered Missions}: We considered 8 random mission with 20 targets that are located uniformly in
one case and in 9 clusters in \ a second case. The goal here is to
investigate the contribution of the sparsity factor $\zeta_{i}$ in
\eqref{zeta}. We have again used $\lambda_{i}\sim U(2,12)$ and $D_{i}=300$. We consider a case with $\gamma=0$ which
eliminates the effect of $\zeta_{i}$ and a second case with $\gamma=0.3$ and $K=5$ in \eqref{zeta}. The results in table \ref{sparsitytable} indicate that in the clustered missions rewards are improved by about
$24\%$ whereas in the uniform cases the reward is unaffected on average.
\begin{table}[ptb]

\caption{Effect of the sparsity factor $\zeta_{i}$ in clustered missions}
\centering
\resizebox{0.8\columnwidth}{!}{
\begin{tabular}{|p{1.1cm}|p{1cm}|p{1.4cm}|p{1cm}|p{1.4cm}|}
\hline
\multirow{2}{*}{Mission $\#$} & \multicolumn{2}{|c|}{$\gamma=0$} & \multicolumn{2}{|c|}{$\gamma=0.3$}\\\cline{2-5}
&Reward & Travel Time & Reward & Travel Time\\
\hline
1  &     40.62    &      552   &      61.9    &      413 \\\hline
2  &     64.89    &      447   &     64.64    &      420 \\\hline
3  &     35.24    &      471   &      63.8    &      461\\\hline
4  &     63.78    &      465   &     64.64    &      478\\\hline
5  &     25.42    &      493   &      26.5    &      449\\\hline
6  &        22    &      454   &        22    &      454\\\hline
7  &      44.1    &      458   &     46.84    &      449\\\hline
8 &      34.26   &       466   &      61.21   &      472\\\hline
\bf Average & 41.29 &    475   &      51.44   &      449 \\\hline
\end{tabular}}\label{sparsitytable}\end{table}
\section{Conclusions and Future Work}
In this work a new CRH controller was developed for solving cooperative multi-agent problems in uncertain environments using the framework of the previous work in \cite{Li2006}. We overcame several limitations of the controller developed in \cite{Li2006}, including agent trajectory instabilities and inaccurate estimation of a reward-to-go function while improving the overall performance.  The event-driven CRH controller is developed to solve the MRCP, where multiple agents cooperate to maximize the total reward collected from a set of stationary targets in the mission space. The mission environment is uncertain, for example targets can appear at random times and agents might have a limited sensing range. The controller sequentially solves optimization problems over a planning horizon and executes the control for a shorter action horizon, where both are defined by certain events associated with new information becoming available. Unlike the earlier CRH controller, the feasible control set is finite instead of an infinite dimensional set. In the numerical comparisons, we showed that the new CRH controller has a better performance than the original one. In future work, the same framework will be applied to problems such as data harvesting where each target is generating data that should be collected and delivered to the base. Here the base will act as a target with dynamic reward. Also the new CRH controller can be extended into a decentralized version where each agent is responsible for calculating its own control.
\appendix
%\begin{spacing}{0.9}
\textbf{Proof of Lemma \ref{MRCPlem1}}
%\begin{proof}
\label{proof:MRCPlem1} From the definition of $\eta_{i}(x,t)$ in \eqref{eta}
and $\mathcal{C}_{l,j}(t_{k},H_{k})$ in \eqref{cp} we have:
\begin{equation}
\mathit{d}(\mathcal{C}_{l,j}(t_{k},H_{k}),\mathbf{y}_{l})\leq\mathit{d}%
(\mathbf{x},\mathbf{y}_{l}),\mbox{   }\forall\mathbf{x}\in\mathcal{F}%
_{j}(t_{k},H_{k})\label{dcp}%
\end{equation}
Dividing both sides by $\lambda_{l}D_{l}^{-1}$ and adding $\zeta_{l}%
(t_{k}+H_{k})$ we get, for all $\mathbf{x}\in\mathcal{F}_{j}(t_{k},H_{k})$,
\begin{equation}
\eta_{l}(\mathcal{C}_{l,j}(t_{k},H_{k}),t_{k}+H_{k})\leq\eta_{l}%
(\mathbf{x},t_{k}+H_{k})\label{subproof}%
\end{equation}
To prove the forward lemma statement, we use a contradiction argument and
assume there exists a target $r$ such that
\[
\eta_{l}(\mathcal{C}_{l,j}(t_{k},H_{k}),t_{k}+H_{k})>\eta_{r}(\mathcal{C}%
_{l,j}(t_{k},H_{k}),t_{k}+H_{k})
\]
Using \eqref{subproof}, we get $\eta_{r}(\mathcal{C}_{l,j}(t_{k},H_{k}%
),t_{k}+H_{k})<\eta_{l}(\mathbf{x},t_{k}+H_{k})$ for all $\mathbf{x}%
\in\mathcal{F}_{j}(t_{k},H_{k})$. This implies that there exists no
$\mathbf{x}\in\mathcal{F}_{j}(t_{k},H_{k})$ such that $l=\argmin_{i}\eta
_{i}(\mathbf{x},t_{k}+H_{k})$. Therefore, $l$ cannot be an active target,
which contradicts the assumption, hence \eqref{MRCPlemma1statement} is
true.\newline To prove the reverse statement, we assume that
\eqref{MRCPlemma1statement} holds for any $i\in\mathcal{T}_{k}$, i.e.,
\[
\eta_{l}(\mathcal{C}_{l,j}(t_{k},H_{k}),t_{k}+H_{k})<\eta_{i}(\mathcal{C}%
_{l,j}(t_{k},H_{k}),t_{k}+H_{k})
\]
By the definition of active targets (\ref{setpi}), we then know that $l$ is an
active target for agent $j$ at time $t_{k}$. \qed
\newline\textbf{Proof of Lemma
\ref{MRCPlem2}}
%\begin{proof}
\label{proof:MRCPlem2} The active target set creates a partition of the set
$\mathcal{F}_{j}(t_{k},H_{k})$ where each subset is an arc in a Euclidean
mission space. For an active target $l\in S_{j}(t_{k},H_{k})$, let the $l$th
arc be ${\mathcal{F}}_{j}^{l}(t_{k},H_{k})\subset\mathcal{F}_{j}(t_{k},H_{k}%
)$. For each ${\mathcal{F}}_{j}^{l}(t_{k},H_{k})$, we prove that the heading
$\mathbf{v}^{\ast}=v(\mathcal{C}_{l,1}(t_{k},H_{k}))$ satisfies, for all
$\mathbf{x}\in{\mathcal{F}}_{j}^{l}(t_{k},H_{k})$:
\[%
\begin{split}
J_{\mathbf{I}}(\mathbf{v}^{\ast},t_{k},H_{k})+ &  J_{\mathbf{A}}%
(\mathbf{v}^{\ast},t_{k},H_{k})>\\
&  J_{\mathbf{I}}(v(\mathbf{x}),t_{k},H_{k})+J_{\mathbf{A}}(v(\mathbf{x}%
),t_{k},H_{k})
\end{split}
\]
There are two possible cases:\newline\emph{Case 1}: $\mathbf{y}_{l}%
\in\mathcal{F}_{1}(t_{k},H_{k})$. This means $d(\mathbf{y}_{l},\mathbf{x}%
_{1}(t))=H_{k}$. Also, from \eqref{cp}, this guarantees that $\forall
r\in\mathcal{T}_{k}$:
\[
q_{r}(\mathcal{C}_{r,1}(t_{k}+H_{k}))=\left\{
\begin{array}
[c]{rl}%
1 & \mbox{  if $r=l$}\\
0 & \text{otherwise}%
\end{array}
\right.
\]
Setting $\tilde{\tau}_{r}(\mathbf{v}^{\ast},t_{k},H_{k}))=\tilde{\tau}%
_{r}^{\ast}$, we have
\begin{align}
J(\mathbf{v}^{\ast},t_{k},H_{k})=&J_{\mathbf{I}}(\mathbf{v}^{\ast}%
,t_{k},H_{k})+J_{\mathbf{A}}(\mathbf{v}^{\ast},t_{k},H_{k})\nonumber\\
= &  \lambda_{l}\phi_{l}(t_{k}+H_{k})+\sum_{r=1}^{M_{k+1}}\lambda_{r}\phi
_{r}(\tilde{\tau}_{r}^{\ast})q_{l}(\mathbf{x}_{1}(\tilde{\tau}_{r}^{\ast
}))\nonumber
\end{align}
Here, $M_{k+1}=|\mathcal{T}_{k+1}|$ and $\mathcal{T}_{k+1}=\mathcal{T}%
_{k}-\{l\}$ since reward $l$ will be already collected at time $t_{k}+H_{k}$.
The estimated visit time $\tilde{\tau}_{r}^{\ast}$ is determined based on
a tour ${\boldsymbol{\theta}}$ that starts at point $\mathbf{y}_{l}$. Now let
us calculate the objective function for any other heading $v(\mathbf{x})$
where $\mathbf{x}\in{\mathcal{F}}_{1}^{l}(t_{k},H_{k})$. Setting $\tilde{\tau
}_{r}(v(\mathbf{x}),t_{k},H_{k}))=\tilde{\tau}_{r}$,
\begin{align}
J(v(\mathbf{x}),t_{k},H_{k})=&J_{\mathbf{I}}(v(\mathbf{x}),t_{k}%
,H_{k})+J_{\mathbf{A}}(v(\mathbf{x}),t_{k},H_{k})\nonumber\\
= &  0+\sum_{r=1}^{M_{k+1}^{\prime}}\lambda_{r}\phi_{r}(\tilde{\tau}_{r}%
)q_{l}(\mathbf{x}_{1}(\tilde{\tau}_{r}))\nonumber
\end{align}
since $\mathbf{x}\neq\mathcal{C}_{l,1}(t_{k},H_{k})$ so that $q_{r}%
(\mathbf{x})=0$ for all $r\in\mathcal{T}_{k}$. The aggregated tour is
determined over the set $\mathcal{T}_{k+1}^{\prime}=\mathcal{T}_{k}$ sarting
at $\mathbf{x}\in{\mathcal{F}}_{1}^{l}(t_{k},H_{k})$. By definition, the
target with the least travel cost from point $\mathbf{x}$ is the active target
$l$ and this is the first target in the tour. The rest of the tour consists of
targets in $\mathcal{T}_{k+1}-\{{l\}}$ starting at $\mathbf{y}_{l}$. Let us
call this tour ${\boldsymbol{\theta}}^{\prime}$. Since in both tours
${\boldsymbol{\theta}}$ and ${\boldsymbol{\theta}}^{\prime}$ the starting
point and the set of available targets are the same, the order of targets
will be identical and we have ${\boldsymbol{\theta}}^{\prime}%
=\{l,{\boldsymbol{\theta}}\}$. The visit times in ${\boldsymbol{\theta}}$
are given by
\[
\tilde{\tau}_{{\boldsymbol{\theta}}_n}^{\ast}=t_{k}+H_{k}+\sum
\limits_{i=1}^{n-1}\mathit{d}(\mathbf{y}_{{\boldsymbol{\theta}}_i}%
,\mathbf{y}_{{\boldsymbol{\theta}}_{i+1}})
\]
In ${\boldsymbol{\theta}}^{\prime}$, the visit time for target
${\boldsymbol{\theta}}^{\prime}_1=l$ is: $\tilde{\tau}_{{\boldsymbol{\theta}
}^{\prime}_1}=t_{k}+H_{k}+\mathit{d}(\mathbf{x},\mathbf{y}_{l})$. For the
rest of the targets, with $1<n\leq M_{k+1}^{\prime}$,
\[
\tilde{\tau}_{{\boldsymbol{\theta}}^{\prime}(n),1}=t_{k}+H_{k}+\mathit{d}%
(\mathbf{x},\mathbf{y}_{l})+\sum\limits_{i=1}^{n-1}\mathit{d}(\mathbf{y}%
_{{\boldsymbol{\theta}}^{\prime}_i},\mathbf{y}_{{\boldsymbol{\theta}}%
^{\prime}_{i+1}})
\]
For all $1<n\leq M_{k+1}$, we have ${\boldsymbol{\theta}}^{\prime
}_{n+1}={\boldsymbol{\theta}}_n$ and $\tilde{\tau}_{{\boldsymbol{\theta}%
}^{\prime}_{n+1}}>\tilde{\tau}_{{\boldsymbol{\theta}}_n}$. By assumption,
for all $i\in\mathcal{T}$, $\phi_{i}(t)$ is non-increasing, therefore $\phi_{{\boldsymbol{\theta}}^{\prime}_{n+1}}(\tilde{\tau
}_{{\boldsymbol{\theta}}^{\prime}_{n+1}})\leq\phi_{{\boldsymbol{\theta}}%
_{n}}(\tilde{\tau}_{{\boldsymbol{\theta}}_n})$, and it follows that
\[%
\begin{split}
&  \lambda_{l}\phi_{l}(t_{k}+H_{k}+\mathit{d}(\mathbf{x},\mathbf{y}_{l}%
))+\sum_{n=2}^{M_{k+1}^{\prime}}\lambda_{{\boldsymbol{\theta}}^{\prime}%
_n}\phi_{{\boldsymbol{\theta}}^{\prime}_n}(\tilde{\tau}_{{\boldsymbol{\theta
}}^{\prime}_n})\\
&  \leq\lambda_{l}\phi_{l}(t_{k}+H_{k})+\sum_{n=1}^{M_{k}+1}\lambda
_{{\boldsymbol{\theta}}_n}\phi_{{\boldsymbol{\theta}}_n}(\tilde{\tau
}_{{\boldsymbol{\theta}}_n})
\end{split}
\]
The right-hand-side above is $J(\mathbf{v}^{\ast},t_{k},H_{k})$ and the
left-hand-side is $J(v(\mathbf{x}),t_{k},H_{k})$, so we have proved that for
any $\mathbf{x}\in{\mathcal{F}}_{1}^{l}(t_{k},H_{k})$, $\mathbf{x}%
\neq\mathcal{C}_{l,1}(t_{k},H_{k})$ we have $J(v(\mathbf{x}),t_{k},H_{k})\leq
J(\mathbf{v}^{\ast},t_{k},H_{k})$.\newline\emph{Case 2}: $\mathbf{y}%
_{l}\not \in \mathcal{F}_{1}(t_{k},H_{k})$. In this case, for any point
$\mathbf{x}\in{\mathcal{F}}_{j}^{l}(t_{k},H_{k})$ we have a zero immediate
reward. Thus, only the rewards-to-go need to be compared. Using
\eqref{headingactive}, for any $\mathbf{x}\in{\mathcal{F}}_{j}^{l}(t_{k}%
,H_{k})$ we know the aggregation tour $\boldsymbol{\theta}$ for any point
$\mathbf{x}$ starts with target $l$ and the rest of it would also be the same.
Similarly, let us assume $\boldsymbol{\theta}$ is the tour for $\mathbf{v}%
^{\ast}$ and $\boldsymbol{\theta}^{\prime}$ is the tour for any other point
$\mathbf{x}$. The estimated visit times for $\boldsymbol{\theta}$ are:
\[
\tilde{\tau}_{{\boldsymbol{\theta}}_n}^{\ast}=t_{k}+H_{k}+\mathit{d}%
(\mathbf{y}_{l},\mathcal{C}_{l,1}(t_{k},H_{k}))+\sum\limits_{i=1}%
^{n-1}\mathit{d}(\mathbf{y}_{{\boldsymbol{\theta}}_i},\mathbf{y}%
_{{\boldsymbol{\theta}}_{i+1}})
\]
and for $\boldsymbol{\theta}^{\prime}$:
\[
\tilde{\tau}_{{\boldsymbol{\theta}}_n}^{\ast}=t_{k}+H_{k}+\mathit{d}%
(\mathbf{y}_{l},\mathbf{x})+\sum\limits_{i=1}^{n-1}\mathit{d}(\mathbf{y}%
_{{\boldsymbol{\theta}}_i},\mathbf{y}_{{\boldsymbol{\theta}}_{i+1}})
\]
By the definition in \eqref{cp}, $\mathcal{C}_{l,1}(t_{k},H_{k}))$ is on the
shortest path from $\mathbf{x}_{j}(t_{k})$ to $\mathbf{y}_{l}$, i.e.,
$\tilde{\tau}_{{\boldsymbol{\theta}}^{\prime}_n}>\tilde{\tau}%
_{{\boldsymbol{\theta}}_n}$. Again, with $\phi_{i}(t)$ being non-increasing
we have $\phi_{{\boldsymbol{\theta}}^{\prime}_n}(\tilde{\tau}%
_{{\boldsymbol{\theta}}^{\prime}_n})\leq\phi_{{\boldsymbol{\theta}}%
_n}(\tilde{\tau}_{{\boldsymbol{\theta}}_n})$, which implies
$J(v(\mathbf{x}),t_{k},H_{k})\leq J(\mathbf{v}^{\ast},t_{k},H_{k})$.

We have thus proved the lemma statement that the optimal heading of the agent
is one of the direct headings towards an active target. \qed
\newline\textbf{Proof
of Theorem \ref{MRCPtheorem1}} \label{Proof:MRCPtheorem1} In the multi-agent mission, calculating the immediate reward and reward-to-go in \eqref{JI} and \eqref{partitionJK} for each
agent is like a one-agent mission limited to its own target subset
$\mathcal{T}_{k,j}$. Therefore, the result follows directly from Lemma
\ref{MRCPlem2}. \qed
\newline\textbf{Proof of Theorem \ref{twotargettheorem}}
%\begin{proof}
\label{proof:twotargettheorem} We assume WLOG that $\mathit{d}(\mathbf{x}%
,\mathbf{y}_{1})<\mathit{d}(\mathbf{x},\mathbf{y}_{2})$ so that at time
$t_{k}$ we have $H_{k}=\mathit{d}(\mathbf{x},\mathbf{y}_{1})$. This implies
that target 1 is always an active target (the travel cost of target 1 at time
$t_{k+1}=t_{k}+H_{k}$ is equal to 0). Recalling \eqref{cp} and setting $\mathcal{C}%
_{2,1}\mathcal{=}\mathcal{C}_{2,1}(t_{k},H_{k})$, we have $\mathit{d}%
(\mathbf{x},\mathbf{y}_{1})=\mathit{d}(\mathbf{x},\mathcal{C}_{2,1})=H_{k}$.
This results in:
\begin{equation}
\mathit{d}(\mathbf{x},\mathbf{y}_{2})=\mathit{d}(\mathbf{x},\mathbf{y}%
_{1})+\mathit{d}(\mathbf{y}_{2},\mathcal{C}_{2,1})\label{cpdifequation}%
\end{equation}
From Lemma \ref{MRCPlem1}, target 2 is an active target if and only if
$\eta_{2}(\mathcal{C}_{2,1},t_{k}+H_{k})\leq\eta_{1}(\mathcal{C}_{2,1}%
,t_{k+1})$. Therefore, from \eqref{eta}, target 2 is an active target if
and only if:
\[
\frac{\mathit{d}(\mathcal{C}_{2,1},\mathbf{y}_{2})}{\lambda_{2}D_{2}^{-1}}%
\leq\frac{\mathit{d}(\mathcal{C}_{2,1},\mathbf{y}_{1})}{\lambda_{1}D_{1}^{-1}}%
\]
which is rewritten as:
\[
\frac{\lambda_{1}}{D_{1}}{\mathit{d}(\mathcal{C}_{2,1},\mathbf{y}_{2})}%
\leq\frac{\lambda_{2}}{D_{2}}{\mathit{d}(\mathcal{C}_{2,1},\mathbf{y}_{1})}%
\]
We now consider two possible cases regarding target 2. First, assume target 2
is not an active target, i.e.,
\begin{equation}
\frac{\lambda_{1}}{D_{1}}{\mathit{d}(\mathcal{C}_{2,1},\mathbf{y}_{2})}%
>\frac{\lambda_{2}}{D_{2}}{\mathit{d}(\mathcal{C}_{2,1},\mathbf{y}_{1}%
)}\label{targe2notactive}%
\end{equation}
%We will prove the optimal path is  ${\bf x}_1(t)\rightarrow {\bf y}_1\rightarrow {\bf y}_2$.\\
Starting with the trivial inequality:
\begin{equation}
0>\frac{-\lambda_{1}}{D_{1}}\big[\mathit{d}(\mathcal{C}_{2,1},\mathbf{y}%
_{2})+\mathit{d}(\mathcal{C}_{2,1},\mathbf{y}_{1})\big]\nonumber
\end{equation}
add $\frac{\lambda_{2}}{D_{2}}\big[\mathit{d}(\mathcal{C}_{2,1},\mathbf{y}%
_{1})\big]$ to both sides and use \eqref{targe2notactive} to get:
\[%
\begin{split}
&  \frac{\lambda_{1}}{D_{1}}\big[\mathit{d}(\mathcal{C}_{2,1},\mathbf{y}%
_{2})\big]>\frac{\lambda_{2}}{D_{2}}\big[\mathit{d}(\mathcal{C}_{2,1}%
,\mathbf{y}_{1})\big]>\\
&  \frac{-\lambda_{1}}{D_{1}}\big[\mathit{d}(\mathcal{C}_{2,1},\mathbf{y}%
_{2})\big]+(\frac{\lambda_{2}}{D_{2}}-\frac{\lambda_{1}}{D_{1}})\mathit{d}%
(\mathcal{C}_{2,1},\mathbf{y}_{1})\big]
\end{split}
\]
Adding the positive quantity of $\frac{\lambda_{2}}{D_{2}}\big[\mathit{d}%
(\mathcal{C}_{2,1},\mathbf{y}_{2})\big]$ to both sides and invoking the
triangle inequality:
\[%
\begin{split}
&  (\frac{\lambda_{1}}{D_{1}}+\frac{\lambda_{2}}{D_{2}})\big[\mathit{d}%
(\mathcal{C}_{2,1},\mathbf{y}_{2})\big]\\
&  >(\frac{\lambda_{2}}{D_{2}}-\frac{\lambda_{1}}{D_{1}})\big[\mathit{d}%
(\mathcal{C}_{2,1},\mathbf{y}_{2})\big]+(\frac{\lambda_{2}}{D_{2}}%
-\frac{\lambda_{1}}{D_{1}})\big[\mathit{d}(\mathcal{C}_{2,1},\mathbf{y}%
_{1})\big]\\
&>(\frac{\lambda_{2}}{D_{2}}-\frac{\lambda_{1}}{D_{1}})\mathit{d}%
(\mathbf{y}_{1},\mathbf{y}_{2})
\end{split}
\]
Rearranging the last inequality and using \eqref{cpdifequation} results in:
\begin{equation}%
\begin{split}
&  \frac{\lambda_{1}}{D_{1}}\big[\mathit{d}(\mathbf{x},\mathbf{y}%
_{2})+\mathit{d}(\mathbf{y}_{2},\mathbf{y}_{1})\big]+\frac{\lambda_{2}}{D_{2}%
}\mathit{d}(\mathbf{x},\mathbf{y}_{2})\\
&  >\frac{\lambda_{1}}{D_{1}}\mathit{d}(\mathbf{x},\mathbf{y}_{1}%
)+\frac{\lambda_{2}}{D_{2}}\big[\mathit{d}(\mathbf{x},\mathbf{y}%
_{1})+\mathit{d}(\mathbf{y}_{2},\mathbf{y}_{1})\big]
\end{split}
\label{finallemma3}%
\end{equation}
which  is the same as \eqref{Twotargetanalytic} implying that path $\theta
_{1}=(1,2)$\textbf{ }is optimal, i.e., the CRH controller finds the optimal
path.

Next, assume that target 2 is also an active target along with target 1. Let
$u_{1}$ and $u_{2}$ be the headings for target 1 and 2 respectively, i.e.,
$\mathbf{x}_{1}(t_{k+1},u_{1})=\mathbf{y}_{1}$ and $\mathbf{x}_{1}%
(t_{k+1},u_{2})=\mathcal{C}_{2,1}$, The objective function of the CRH
controller under $u_{1}$ and $u_{2}$ is:
\[%
\begin{split}
J(u_{1},t_{k},H_{k})&=J_{\mathbf{I}}(u_{1},t_{k},H_{k})+J_{\mathbf{A}}%
(u_{1},t_{k},H_{k})\\
&  =\lambda_{1}\phi_{1}(t_{k+1})+\lambda_{2}\phi_{2}(t_{k+1}+\mathit{d}(\mathbf{y}_{1},\mathbf{y}_{2}))
\end{split}
\]%
\[%
\begin{split}
J(u_{2},t_{k},H_{k})&=J_{\mathbf{I}}(u_{2},t_{k},H_{k})+J_{\mathbf{A}}%
(u_{2},t_{k},H_{k})\\
&=0+\big[\lambda_{2}\phi_{2}\big(t_{k+1}+\mathit{d}(\mathcal{C}%
_{2,1},\mathbf{y}_{2})\big)\\&+\lambda_{1}\phi_{1}\big(t_{k+1}  +\mathit{d}(\mathcal{C}_{2,1},\mathbf{y}_{2})+\mathit{d}(\mathbf{y}%
_{1},\mathbf{y}_{2})\big)
\end{split}
\]
Note that in order to evaluate the objective function for $u_{2}$ we find a
tour starting at point $\mathcal{C}_{2,1}$ which goes to the target with
minimum travel cost. However, for target 2 to be active at $t_{k}$ it has to
have the smallest travel cost at that point, which results in $J_{\mathbf{A}%
}(u_{2},t_{k},H_{k})$ to be the reward of going to target 2 and then target 1.
We can see that using the reward of each path from \eqref{R12} and \eqref{R21}
we can write:
\[
J(u_{1},t_{k},H_{k})=R_{(1,2)},\qquad J(u_{2},t_{k},H_{k})=R_{(2,1)}%
\]
Thus, the objective function of the CRH controller under $u_{1}$ and $u_{2}$
is identical to the corresponding path rewards. Hence, the CRH controller
selects the correct optimal heading at $t_{k}$. \qed
%\end{proof}
%\end{spacing}
\bibliographystyle{ieeetr}

\end{document}